\newcommand{\maketextstyle}{\textstyle} %

\documentclass[letterpaper, 10 pt, conference]{ieeeconf}
\IEEEoverridecommandlockouts %
\usepackage{mypreamble}

\title{\bf \LARGE Nonlinear Bandwidth and Bode Diagrams based on \\ Scaled Relative Graphs}
\author{Julius P.~J. Krebbekx$^1$, Roland Tóth$^{1,2}$, Amritam Das$^1$
\thanks{$^1$Control Systems group, Department of Electrical Engineering,  Eindhoven University of Technology, The Netherlands.}
\thanks{$^2$Systems and Control Lab, HUN-REN Institute for Computer Science and Control, Budapest, Hungary. 
E-mail: {\tt\small \{j.p.j.krebbekx, r.toth, am.das\}@tue.nl}}
}

\date{\today}

\begin{document}

\maketitle

\begin{abstract}
    \emph{Scaled Relative Graphs} (SRGs) provide a novel graphical frequency-domain method for the analysis of \emph{Nonlinear} (NL) systems. In this paper, we restrict the SRG to particular input spaces to compute frequency-dependent incremental gain bounds for nonlinear systems. This leads to a NL generalization of the Bode diagram, where the sinusoidal, harmonic, and subharmonic inputs are considered separately. When applied to the analysis of the NL loop transfer and sensitivity, we define a notion of bandwidth for both the open-loop and closed-loop, compatible with the \emph{Linear Time-Invariant} (LTI) definitions. We illustrate the power of our method on the analysis of a DC motor with a parasitic nonlinearity and verify our results in simulations.
\end{abstract}

\section{Introduction}

In the case of \emph{Linear Time-Invariant} (LTI) systems, graphical analysis using the Nyquist diagram~\cite{nyquistRegenerationTheory1932} and the Bode diagram~\cite{bodeRelationsAttenuationPhase1940} forms the cornerstone of control engineering. They are easy to use and allow for intuitive analysis and controller design methods. However, it is unclear how to generalize graphical frequency-domain methods to nonlinear system analysis and controller design.

When pushing the performance of dynamical systems through control design, the \emph{Nonlinear} (NL) dynamics start to play a crucial role. Inspired by the success of analyzing performance through Bode diagrams in the LTI case, e.g. for mixed sensitivity shaping~\cite{skogestadMultivariableFeedbackControl2010}, there have been multiple attempts to generalize the Bode diagram to the NL case. The \emph{Describing Function} (DF)~\cite{krylovIntroductionNonlinearMechanics1947}, which is approximative, has been the first step in this direction, and has been successfully applied in practice~\cite{pogromskyFrequencyDomainPerformance2014,heertjesVariableGainMotion2016}. Most modern and non-approximative methods focus on sinusoidal inputs~\cite{pavlovFrequencyDomainPerformance2007,heertjesVariableGainMotion2016}, but other works also consider higher-order harmonics in the input~\cite{rijlaarsdamComparativeOverviewFrequency2017} or specific periodic inputs~\cite{pavlovSteadystatePerformanceOptimization2013}.

The \emph{Scaled Relative Graph} (SRG) \cite{ryuScaledRelativeGraphs2022}, proposed in~\cite{chaffeyGraphicalNonlinearSystem2023}, is a new graphical method to analyze nonlinear feedback systems. It is an exact method, and it is intuitive because of its close connection to the Nyquist diagram. Moreover, SRG analysis can provide performance bounds in terms of (incremental) $L_2$-gain. Originally, SRG analysis was developed for stable \emph{Single-Input Single-Output} (SISO) systems only and has recently been extended to include unstable elements in the loop~\cite{krebbekxScaledRelativeGraph2024,krebbekxScaledRelativeGraph2025,chenSoftHardScaled2025}, and \emph{Multiple-Input Multiple-Output} (MIMO) systems~\cite{krebbekxGraphicalAnalysisNonlinear2025}, which may be non-square. While practical stability analysis is an important outcome of the SRG framework, our aim is to take one step further and establish a non-approximative frequency-domain performance shaping tool for NL systems.

In this paper, we aim to establish such a tool for SISO NL systems with the property that they preserve the periodicity of the input. Our approach is based on a \emph{Linear Fractional Representation} (LFR) of the system, where the nonlinearity can be static or dynamic, making our method highly flexible. The core idea is that we can compute the incremental gain of the system for sets of input signals with a common period, where the SRG is used as the main computational tool. By evaluating these gains over a grid of frequencies, one obtains a Bode diagram for the NL system. The NL Bode diagram is used to define the NL bandwidth for the loop transfer and sensitivity, compatible with the LTI definition. 

Our method can reproduce most aspects of existing work, such as sinusoidal input Bode diagrams. Also, our results are more general since we are able to compute the gain for any frequency and \emph{arbitrary} higher harmonics, but amplitude-dependent results, compared to, e.g.~\cite{pavlovUniformOutputRegulation2006}, are not yet reflected in our approach. An entirely novel aspect of our work is the gain for signals with \emph{subharmonics}, which allows low-frequency analysis of sensitivity functions, going beyond sinusoidal inputs. Finally, all our computations are based on modular interconnections of LTI and NL input/output operators, making the method amenable to data-driven techniques. 

This paper is structured as follows. In Section~\ref{sec:preliminaries}, we present the required preliminaries. In Section~\ref{sec:freq_domain_analysis_of_incremental_systems}, we derive a method to analyze NL systems in the frequency domain, leading to our definition of the NL Bode plot. We show how SRGs are used to compute the NL Bode diagrams in Section~\ref{sec:nl_bode_using_srg}, focusing on the loop transfer and the sensitivity. Finally, we apply our results to a practical design example in Section~\ref{sec:example} and present our conclusions in Section~\ref{sec:conclusion}.

\section{Preliminaries}\label{sec:preliminaries}

\subsection{Notation and Conventions}

Let $\R$ and $\C$ denote the field real and complex numbers, respectively, with $\R_{>0} = (0, \infty)$, $\R_{\geq 0} = [0, \infty)$ and $\C_{\mathrm{Im} \geq 0}= \{ a+ jb \mid \, a \in \R, \, b \in \R_{\geq 0} \}$, where $j$ is the imaginary unit. We denote the complex conjugate of $z = a + jb \in \C$ as $\bar{z} = a-jb$. Let $\mathcal{L}$ denote a Hilbert space, with inner product $\inner{\cdot}{\cdot}_\mathcal{L} : \mathcal{L} \times \mathcal{L} \to \C$ and norm $\norm{x}_\mathcal{L} := \sqrt{\inner{x}{x}_\mathcal{L}}$.  For sets $A, B \subseteq \C$, the sum and product sets are defined as $A+B:= \{ a+b \mid a\in A, b\in B\}$ and $AB:= \{ ab \mid a\in A, b\in B\}$, respectively. A closed disk in the complex plane is $D_r(x) = \{ z \in \C \mid |z-x| \leq r \}$. Denote $D_{[\alpha, \beta]}$ the disk in $\C$ centered on $\R$ which intersects $\R$ in $[\alpha, \beta]$. The radius of a set $\mathcal{C} \subseteq \C$ is defined by $\rmin(\mathcal{C}) := \min_{r>0} : \mathcal{C} \subseteq D_r(0)$. The distance between two sets $\mathcal{C}_1,\mathcal{C}_2 \subseteq \C_\infty$ is defined as $\dist(\mathcal{C}_1,\mathcal{C}_2) := \inf_{z_1 \in \mathcal{C}_1, z_2 \in \mathcal{C}_2} |z_1-z_2|$, where $|\infty-\infty|:=0$.

\subsection{Signals, Systems and Stability}
Since this work focuses on SISO continuous-time systems, the Hilbert space of particular interest is $L_2(\mathbb{F}):= \{ f:\R_{\geq 0} \to \mathbb{F} \mid \norm{f} < \infty \}$, where $\mathbb{F} \in \{\R, \C\}$, the norm is induced by the inner product $\inner{f}{g}:= \int_\mathbb{F} \bar{f}(t) g(t) d t$, $\bar{f}$ denotes the complex conjugate of $f$, and $\mathbb{T}\in \{\R_{\geq 0}, [0,T]\}$ for any $T>0$ is the domain. For brevity, we denote $L_2(\R_{\geq 0}, \mathbb{F})$ as $L_2(\mathbb{F})$, $L_2(\R_{\geq 0}, \R)$ as $L_2$ and, if the time domain is finite, $L_2([0,T],\R)$ as $L_2[0,T]$.

For any $T \in \R_{\geq 0}$, define the truncation operator $P_T : L_2(\mathbb{F}) \to L_2(\mathbb{F})$ as $(P_T u)(t) = 0$ for all $t>T$, else $(P_T u)(t) = u(t)$.
The extension of $L_2(\mathbb{F})$, see Ref.~\cite{desoerFeedbackSystemsInputoutput1975}, is defined as 
\begin{equation*}
    \Lte(\mathbb{F}) := \{ u : \R_{\geq 0} \to \mathbb{F} \mid \norm{P_T u} < \infty \text{ for all } T \in \R_{\geq 0} \}.
\end{equation*}
The space $\Lte(\R)$, which we denote from now on as $\Lte$, will be the most frequently used space of signals. Note that the extension is particularly useful since it includes periodic signals, which are otherwise excluded from $L_2$. 

Periodic signals $v\in \Lte$ can also be viewed as elements in $L_2[0,T]$, where $T$ is the period of the signal, i.e. $v(t)=v(t+T)$ for all $t\in [0,\infty)$. Any $u \in L_2[0,T]$ can be written in Fourier series as $u(t) = \sum_{k \in \Z} \hat{u}_k e^{2 \pi j k t/T}$
where $\hat{u}_k \in \C$ are the Fourier coefficients. The \emph{Root-Mean-Square} (RMS) norm of the signal is defined as %
\begin{equation}\label{eq:rms_fourier_L2}
    \maketextstyle \norm{u}_\mathrm{RMS} := \sqrt{\sum_{k \in \Z} |\hat{u}_k|^2} = 1/\sqrt{T} \norm{u}_{L_2[0,T]}.
\end{equation}

Systems are modeled as operators $R: \Lte \to \Lte$. A system is said to be causal if it satisfies $P_T (Ru) = P_T(R(P_Tu))$, i.e., the output at time $t$ is independent of the signal at times greater than $t$. Unless specified otherwise, we will always assume causality.

Given an operator $R$ on $L_2$, the induced incremental norm of the operator is defined (similar to the notation in~\cite{vanderschaftL2GainPassivityTechniques2017}) as 
\begin{equation}\label{eq:incremental_induced_norm}
    \maketextstyle \Gamma(R) := \sup_{u_1, u_2 \in L_2} \frac{\norm{Ru_1-Ru_2}}{\norm{u_1-u_2}}.
\end{equation}

For causal systems, the induced incremental operator norm on $L_2$ carries over to  $\Lte$ since $\norm{P_T(R(P_Tu))} \leq \norm{R(P_Tu)}$ and $P_T u \in L_2$ for all $u \in \Lte$. We define the incremental $L_2$-gain of a causal operator $R : \Lte(\mathbb{F}) \to \Lte(\mathbb{F})$ as $\Gamma(R)$, i.e., the induced incremental operator norm from~\eqref{eq:incremental_induced_norm}. %

\subsection{Complex Geometry}

Let $z_1, z_2 \in \C_{\mathrm{Im} \geq 0}$ where we assume w.l.o.g. that $\mathrm{Re} (z_1) \leq \mathrm{Re} (z_2)$. Denote $\operatorname{Circ}(z_1, z_2)$ the unique circle through $z_1, z_2$ centered on $\R$. Let 
\begin{multline*}
    \operatorname{Arc}_{\operatorname{min}}(z_1, z_2) = \\ \{ z \in \operatorname{Circ}(z_1, z_2) \mid \mathrm{Re} (z_1) \leq \mathrm{Re} (z) \leq \mathrm{Re} (z_2), \mathrm{Im} (z) \geq 0  \}.
\end{multline*} 

\begin{definition}[h-convex]
    A set $S \subseteq \C_{\mathrm{Im} \geq 0}$ is h-convex if 
    \begin{equation*}
        z_1, z_2 \in S \iff \operatorname{Arc}_{\operatorname{min}}(z_1, z_2) \subseteq S.
    \end{equation*}
    Given a set of points $P \subseteq \C_{\mathrm{Im} \geq 0}$, the h-convex hull of $P$ is the smallest set $\tilde{P} \supseteq P$ that is h-convex. We denote the h-convex hull as $\Tilde{P}= \hco (P)$.
\end{definition}

For a set $P \subseteq \C$ that is equal to its complex conjugate $\bar{P} = P$, i.e., is symmetric w.r.t. the real axis, h-convexity can be defined for $P_+ := P \cap \C_{\mathrm{Im} \geq 0}$. In that case, the h-convex hull is defined $\hco(P) = \hco(P_+) \cup \overline{\hco(P_+)}$.

\subsection{Scaled Relative Graphs}\label{sec:srg_definitions}

We now turn to the definition and properties of the SRG as introduced by Ryu et al. in~\cite{ryuScaledRelativeGraphs2022}. We follow closely the exposition of the SRG as given by Chaffey et al. in~\cite{chaffeyGraphicalNonlinearSystem2023}. 

\subsubsection{Definitions}

Let $\mathcal{L}$ be a Hilbert space, and $R : \mathcal{L} \to \mathcal{L}$ an operator. The angle between $u, y\in \mathcal{L}$ is defined as 
\begin{equation}\label{eq:def_srg_angle}
    \maketextstyle \angle(u, y) := \cos^{-1} \frac{\mathrm{Re} \inner{u}{y}}{\norm{u} \norm{y}} \in [0, \pi].
\end{equation}
Given $u_1, u_2 \in \mathcal{U} \subseteq \mathcal{L}$, define the set of complex numbers
\begin{equation*}
    \maketextstyle z_R(u_1, u_2) := \left\{ \frac{\norm{Ru_1-Ru_2}}{\norm{u_1-u_2}} e^{\pm j \angle(u_1-u_2, Ru_1-Ru_2)}  \right\}.
\end{equation*}
The SRG of $R$ over the set $\mathcal{U}$ is defined as
\begin{equation}\label{eq:srg_specific_input_space}
    \maketextstyle \SRG_\mathcal{U} (R) := \bigcup_{u_1, u_2 \in \mathcal{U}} z_R(u_1, u_2).
\end{equation}
If $\mathcal{U}=\mathcal{L}$, we simply write $\SRG(R)$. The radius of the SRG corresponds to the incremental induced gain as $\rmin(\SRG(R)) = \Gamma(R)$ in terms of~\eqref{eq:incremental_induced_norm}.

\subsubsection{Operations on SRGs}\label{sec:operations_on_srgs}
The facts presented here are proven in~\cite[Chapter 4]{ryuScaledRelativeGraphs2022}. 

Inversion of a point $z = re^{j \omega} \in \C$ is defined as the M\"obius inversion $r e^{j \omega} \mapsto (1/r)e^{j \omega}$. We refer the reader to~\cite[Ch. 4.4 and 4.5]{ryuScaledRelativeGraphs2022} for the definitions of the chord and arc property. %

\begin{proposition}\label{prop:srg_calculus}
    Let $0 \neq \alpha \in \R$ and let $R, S$ be arbitrary operators on the Hilbert space $\mathcal{L}$. Then, 
    \begin{enumerate}[label=\alph*.]
        \item\label{eq:srg_calculus_alpha} $\SRG(\alpha R) = \SRG(R \alpha) = \alpha \SRG(R)$,
        \item\label{eq:srg_calculus_plus_one} $\SRG(I + R) = 1 + \SRG(R)$, where $I$ denotes the identity on $\mathcal{L}$,
        \item\label{eq:srg_calculus_inverse} $\SRG(R^{-1}) = (\SRG(R))^{-1} =: \SRG(R)^{-1}$.
        \item\label{eq:srg_calculus_parallel} If at least one of $R, S$ satisfies the chord property, then $\SRG(R + S) \subseteq \SRG(R) + \SRG(S)$.
        \item\label{eq:srg_calculus_series} If at least one of $R, S$ satisfies an arc property, then $\SRG(R S) \subseteq \SRG(R) \SRG(S)$.
    \end{enumerate}
    If the SRGs above contain $\infty$ or are the empty set, the above operations are slightly different, see~\cite{ryuScaledRelativeGraphs2022}. 
\end{proposition}

\subsubsection{Stability analysis}

SRGs serve as a tool to compute the (incremental) $L_2$-gain of a system. We cite the central result from~\cite{chaffeyGraphicalNonlinearSystem2023} (corrected in~\cite{chaffeyHomotopyTheoremIncremental2024}), which considers any system $H_1$ in feedback with another system $H_2$, as displayed in Fig.~\ref{fig:chaffey_thm2}. 

\begin{proposition}\label{thm:chaffey_thm2}
    Let $H_1, H_2$ be operators on $L_{2e}$, where $\Gamma(H_1) < \infty, \Gamma(H_2) < \infty$ and for all $\tau \in (0, 1]$
    \begin{equation*}
        \dist(\SRG(H_1)^{-1}, -\tau \SRG(H_2)) \geq r_m >0,
    \end{equation*}
    and at least one of $H_1, H_2$ obeys the chord property. Then, the system $T=(H_1^{-1}+H_2)^{-1}$ in Fig.~\ref{fig:chaffey_thm2} obeys $\Gamma(T) \leq 1/r_m$.
\end{proposition}

\begin{figure}[t]
    \centering

    \tikzstyle{block} = [draw, rectangle, 
    minimum height=2em, minimum width=2em]
    \tikzstyle{sum} = [draw, circle, scale=0.5, node distance={0.5cm and 0.5cm}]
    \tikzstyle{input} = [coordinate]
    \tikzstyle{output} = [coordinate]
    \tikzstyle{pinstyle} = [pin edge={to-,thin,black}]
    
    \begin{tikzpicture}[auto, node distance = {0.3cm and 0.5cm}]
        \node [input, name=input] {};
        \node [sum, right = of input] (sum) {};
        \node [block, right = of sum] (lti) {$H_1$};
        \node [coordinate, right = of lti] (z_intersection) {};
        \node [output, right = of z_intersection] (output) {}; %
        \node [block, below = of lti] (static_nl) {$H_2$};
    
        \draw [->] (input) -- node {$r$} (sum);
        \draw [->] (sum) -- node {$e$} (lti);
        \draw [->] (lti) -- node [name=z] {$y$} (output);
        \draw [->] (z) |- (static_nl);
        \draw [->] (static_nl) -| node[pos=0.99] {$-$} (sum);
    \end{tikzpicture}
    
    \caption{Block diagram of a general feedback interconnection where $H_1$ and $H_2$ can be LTI or NL static or dynamic blocks.}
    \label{fig:chaffey_thm2}
    \vspace{-1em}
\end{figure}
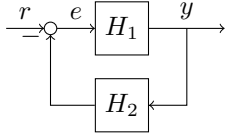

Note that Proposition~\ref{thm:chaffey_thm2} works only in the case of stable open-loop plants $H_1$. In~\cite{krebbekxScaledRelativeGraph2024}, Proposition~\ref{thm:chaffey_thm2} has been extended to the case where $H_1$ is now an unstable LTI operator.

\section{Frequency-Domain analysis of Incrementally Stable Systems}\label{sec:freq_domain_analysis_of_incremental_systems}

To formulate an effective frequency-domain interpretation of NL systems, the core idea of this work is to focus on systems for which periodic inputs lead to periodic outputs. Sufficient conditions for period preservation are incremental stability~\cite{angeliLyapunovApproachIncremental2002}, convergent systems~\cite{pavlovUniformOutputRegulation2006} or fading memory~\cite{boydFadingMemoryProblem1985}, to name a few. In other works, e.g.~\cite{nuijHigherorderSinusoidalInput2006}, period preservation is simply assumed. This property is used to connect the incremental gain of a system to the gain of the periodic output. Particularly, the SRG is used as the direct computational tool to compute this gain for different frequency-dependent input spaces, yielding a Bode diagram for nonlinear systems.

\subsection{The Period Preserving Property}

The fundamental property that a periodic input leads to an eventually periodic output, called the \emph{period preserving property}, is defined as follows. 

\begin{definition}\label{def:T-periodic}
    A signal $x \in \Lte$ is called $T$-periodic for some $T>0$, if $x(t+T)=x(t)$ for all $t \geq 0$.
\end{definition}

\begin{definition}\label{def:period-preserving-property}
    An operator $R : \Lte \to \Lte$ is called period preserving if for every input $u \in \Lte$ that is $T$-periodic, the output $y=Ru$ converges asymptotically to a $T$-periodic signal $\tilde{y}$, i.e. $\forall \epsilon>0 , \exists t>0 : \forall \tau \geq t, \; |y(\tau)-\tilde{y}(\tau)| < \epsilon$.
\end{definition}

We consider one period of the input $u$ and periodic output $\tilde{y}$ as signals in $L_2[0,T]$, and compute their RMS norms using~\eqref{eq:rms_fourier_L2}. The following lemma relates the incremental gain to the RMS norm of a period preserving operator.

\begin{lemma}\label{lemma:incremental_gain_rms_gain_relation}
    Let $R : \Lte \to \Lte$ be causal and period preserving with \mbox{$\Gamma(R) < \infty$}. For any pair $u_1,u_2 \in \Lte$ of $T$-periodic inputs 
    \begin{equation}\label{eq:rms_L2_gain_eq}
        \lim_{t \to \infty} \frac{\norm{P_t(Ru_1 - Ru_2)}}{\norm{P_t(u_1-u_2)}} = \frac{\norm{\Tilde{y_1}-\Tilde{y_2}}_\mathrm{RMS}}{\norm{u_1-u_2}_\mathrm{RMS}},
    \end{equation}
    and consequently
    \begin{equation}\label{eq:rms_gain_bound}
        \sup_{u_1,u_2 \in L_2[0,T]} \frac{\norm{\Tilde{y_1}-\Tilde{y_2}}_\mathrm{RMS}}{\norm{u_1-u_2}_\mathrm{RMS}} \leq \Gamma(R),
    \end{equation}
    allowing to interpret the incremental gain as the RMS gain.
\end{lemma}

\begin{proof}%
    Fix a period $T>0$, $T$-periodic inputs $u_1,u_2$ and corresponding outputs $y_1,y_2$ which have $T$-periodic limits $\tilde{y}_1,\tilde{y}_2$ due to the period preserving property. Take an $\epsilon>0$ and take $t = nT$ for some $n \in \N$ such that for all $\tau \geq t$, $|y_i(\tau)-\tilde{y}_i(t)|< \epsilon =: \epsilon(t)$ holds for $i=1,2$. Since
    \begin{multline*}
        \norm{Ru_1 - Ru_2}_{L_2[0,t+t^2]}\leq \norm{Ru_1 - Ru_2}_{L_2[0,t]} + \\ \norm{\tilde{y}_1-\tilde{y}_2}_{L_2[t,t+t^2]} + \textstyle \sum_{i=1}^2 \norm{Ru_i-\tilde{y}_i}_{L_2[t,t+t^2]}, %
    \end{multline*}
    where $\norm{Ru_1-Ru_2}_{L_2[0,t]} \leq \Gamma(R) \norm{u_1-u_2}_{L_2[0,t]} = \Gamma(R) \norm{u_1-u_2} \sqrt{\tfrac{t}{T}}$, $\norm{\tilde{y}_1-\tilde{y}_1}_{L_2[t,t+t^2]}  = \norm{\tilde{y}_1-\tilde{y_2}} \tfrac{t}{T}$ and $\norm{Ru_i-\tilde{y}_i}_{L_2[t,t+t^2]} \leq \epsilon(t) t$ for $i=1,2$, and $\norm{\cdot}:=\norm{\cdot}_{L_2[0,T]}$ for brevity. As $t=nT \xrightarrow[]{n \to \infty} \infty$, we get
    \begin{multline*}
        \textstyle \frac{\norm{Ru_1-Ru_2}_{L_{2\mathrm{e}}}}{\norm{u_1-u_2}_{L_{2\mathrm{e}}}} = \lim_{t \to \infty} \frac{\norm{Ru_1-Ru_2}_{L_2[0,t+t^2]}}{\norm{u_1-u_2}_{L_2[0,t+t^2]}}  \\ \textstyle \leq  \lim_{t \to \infty} \frac{\Gamma(R) \norm{u_1-u_2} \sqrt{\frac{t}{T}} + \norm{\tilde{y}_1-\tilde{y}_2} \frac{t}{T} + 2\epsilon(t)t}{\norm{u_1-u_2} \sqrt{\frac{t}{T}+(\frac{t}{T})^2}} %
        \\ \textstyle \leq  \lim_{t \to \infty} \frac{\Gamma(R) \norm{u_1-u_2} \sqrt{\frac{T}{t}} + \norm{\tilde{y}_1-\tilde{y}_2}+ 2\epsilon(t)/T}{\norm{u_1-u_2} \sqrt{1+T/t}} = \frac{\norm{\tilde{y}_1-\tilde{y}_2}}{\norm{u_1-u_2}},
    \end{multline*}
    since if $n \to \infty$ in $t=nT$, one has $\epsilon(t) \to 0$.
    Similarly, for the opposite direction it holds
    \begin{multline*}
        \textstyle \norm{\tilde{y}_1-\tilde{y}_2}_{L_2[0,t+t^2]} \leq \norm{\tilde{y}_1-\tilde{y}_2}_{L_2[0,t]} \\ \textstyle + \norm{Ru_1-Ru_2}_{L_2[t,t+t^2]} + \sum_{i=1}^2  \norm{Ru_i-\tilde{y}_i}_{L_2[t,t+t^2]},
    \end{multline*}
    so we can conclude $\frac{\norm{\tilde{y}_1-\tilde{y}_2}}{\norm{u_1-u_2}} \leq \frac{\norm{Ru_1-Ru_2}_{L_{2\mathrm{e}}}}{\norm{u_1-u_2}_{L_{2\mathrm{e}}}} \leq \frac{\norm{\tilde{y}_1-\tilde{y}_2}}{\norm{u_1-u_2}} $ and hence $\frac{\norm{Ru_1-Ru_2}_{L_{2\mathrm{e}}}}{\norm{u_1-u_2}_{L_{2\mathrm{e}}}} = \frac{\norm{\tilde{y}_1-\tilde{y}_2}}{\norm{u_1-u_2}}$, which yields~\eqref{eq:rms_L2_gain_eq} by using $\norm{u}_{L_2[0,T]} = \sqrt{T} \norm{u}_\mathrm{RMS}$. Note that~\eqref{eq:rms_gain_bound} follows immediately from~\eqref{eq:incremental_induced_norm} and~\eqref{eq:rms_L2_gain_eq}.
\end{proof}

This lemma serves as the cornerstone of our frequency-domain analysis of NL systems. Note that~\eqref{eq:rms_gain_bound} is equivalent to the $\Hinf$-norm if $R$ is an LTI operator (see~\cite[A.5.7]{skogestadMultivariableFeedbackControl2010}).

\begin{remark}
    Lemma~\ref{lemma:incremental_gain_rms_gain_relation} can be stated in the \emph{non-incremental} setting by taking $u_2=0$. The resulting gain is the non-incremental gain, and the computations are done using the Scaled Graph (SG) instead. Details on non-incremental analysis and the SG can be found in~\cite{krebbekxScaledRelativeGraph2025,krebbekxGraphicalAnalysisNonlinear2025,vandeneijndenScaledGraphsReset2024,chaffeyGraphicalNonlinearSystem2023}.
\end{remark}

\subsection{Nonlinear Bode Plots}

When an input has period $T$, we call $\omega = 2\pi/T$ the \emph{base harmonic}. The key idea is now to compute the gain in the left-hand side of~\eqref{eq:rms_gain_bound} for a specific space of input signals that corresponds to a given base harmonic. 

\begin{definition}
    For a frequency $\omega \in \R_{>0}$, we define
    \begin{subequations}\label{eqs:input_spaces}
    \begin{align}
        \mathcal{U}_\omega &:= \maketextstyle  \{ u\in \Lte \mid u(t) =  a \sin (\omega t + \phi), \, a,\phi \in \R \}, \\
        \overline{\mathcal{U}}_\omega &:= \maketextstyle \{ u\in \Lte \mid u(t) = \sum_{n \in \Z} \hat u_n e^{j \omega n t} \}, \\
        \underline{\mathcal{U}}_\omega &:= \maketextstyle \{  u\in \Lte \mid u(t) = \sum_{0\neq n \in \Z} \hat u_n e^{j (\omega / n) t}  \},
    \end{align}
    \end{subequations}
    which are called the sinusoidal, harmonic and subharmonic input spaces, respectively. The input space specific gains are defined as
    \begin{subequations}\label{eqs:Gammas}
    \begin{align}
        \Gamma_\omega(R) &= \maketextstyle \sup_{u_1,u_2 \in \mathcal{U}_\omega} \frac{\norm{\tilde{y}_1-\tilde{y}_2}_\mathrm{RMS}}{\norm{u_1-u_2}_\mathrm{RMS}}, \label{eq:Gamma_sinus} \\
        \overline{\Gamma}_\omega(R) &= \maketextstyle \sup_{u_1,u_2 \in \overline{\mathcal{U}}_\omega} \frac{\norm{\tilde{y}_1-\tilde{y}_2}_\mathrm{RMS}}{\norm{u_1-u_2}_\mathrm{RMS}}, \label{eq:Gamma_harmonic} \\
        \textstyle \underline{\Gamma}_\omega(R) &= \maketextstyle \sup_{u_1,u_2 \in \underline{\mathcal{U}}_\omega} \frac{\norm{\tilde{y}_1-\tilde{y}_2}_\mathrm{RMS}}{\norm{u_1-u_2}_\mathrm{RMS}}. \label{eq:Gamma_subharmonic}
    \end{align}
    \end{subequations}
\end{definition}

Note that since $\mathcal{U}_\omega \subseteq \overline{\mathcal{U}}_\omega \subseteq \Lte$ and $\mathcal{U}_\omega \subseteq \underline{\mathcal{U}}_\omega \subseteq \Lte$ for all $\omega \in [0, \infty)$, it holds that 
\begin{equation}\label{eq:sinus_harmonic_gain_inequalities}
    \Gamma_\omega(R) \leq \overline{\Gamma}_\omega(R) \leq \Gamma(R), \quad \Gamma_\omega(R) \leq \underline{\Gamma}_\omega(R) \leq \Gamma(R).
\end{equation}
Moreover, because $\lim_{\omega \to \infty} \underline{\mathcal{U}}_\omega = \lim_{\omega \downarrow 0} \overline{\mathcal{U}}_\omega = \Lte$, we can conclude that 
\begin{equation}\label{eq:harm_subharm_limit}
    \maketextstyle \lim_{\omega \to \infty} \underline{\Gamma}_\omega (R) = \lim_{\omega \downarrow 0} \overline{\Gamma}_\omega(R) = \Gamma(R). 
\end{equation}

The sinusoidal inputs are often used in experiments, and the harmonic inputs correspond to periodic setpoints. The subharmonic inputs occur less frequently in practice, and serve the purpose to probe the low-frequency behavior of the system.

The \emph{NL Bode plot} is obtained by plotting~\eqref{eq:Gamma_sinus}, \eqref{eq:Gamma_harmonic} and \eqref{eq:Gamma_subharmonic} as function of $\omega \in \R_{>0}$ in one graph, analogously to a conventional Bode plot. To the best knowledge of the authors, the \emph{NL Bode plot} provides a novel way to study the frequency-domain behavior of nonlinear operators that are causal and period preserving. We will compare our result to existing methods in Section~\ref{sec:comparison_with_existing_methods}.

\section{Nonlinear Bode Diagrams using Scaled Relative Graphs}\label{sec:nl_bode_using_srg}

\subsection{The Feedback Interconnection for Loop Shaping}\label{sec:lfr_simple_feedback}

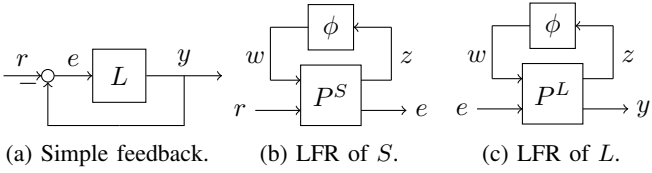
\begin{figure}[t]
    \centering
    \begin{subfigure}[t]{0.32\linewidth}
        \centering
        \tikzstyle{block} = [draw, rectangle, 
        minimum height=2em, minimum width=2em]
        \tikzstyle{sum} = [draw, circle, scale=0.5, node distance={0.5cm and 0.5cm}]
        \tikzstyle{input} = [coordinate]
        \tikzstyle{output} = [coordinate]
        \tikzstyle{pinstyle} = [pin edge={to-,thin,black}]
        
        \begin{tikzpicture}[auto, node distance = {0.3cm and 0.5cm}]
            \node [input, name=input] {};
            \node [sum, right = of input] (sum) {};
            \node [block, right = of sum] (lti) {$L$};
            \node [coordinate, right = of lti] (z_intersection) {};
            \node [output, right = of z_intersection] (output) {}; %
            \node [coordinate, below = of lti] (static_nl) {};
        
            \draw [->] (input) -- node {$r$} (sum);
            \draw [->] (sum) -- node {$e$} (lti);
            \draw [->] (lti) -- node [name=z] {$y$} (output);
            \draw [-] (z) |- (static_nl);
            \draw [->] (static_nl) -| node[pos=0.99] {$-$} (sum);
        \end{tikzpicture}
        \caption{Simple feedback.}
        \label{fig:simple_feedback}
    \end{subfigure}
    \hfill
    \begin{subfigure}[t]{0.32\linewidth}
     \centering
        \begin{tikzpicture}
            \draw[black] (-.4,-.4) rectangle (.4,.4);
            \node at (0,0) {$P^S$};

            \draw[black] (-.3,-.3+.9) rectangle (.3,.3+.9);
            \node at (0,.9) {$\phi$};

            \draw[->] (-1,-.2) -- (-.4,-.2);
            \node at (-1.2,-.2) {$r$};

            \draw[->] (.4,-.2) -- (1,-.2);
            \node at (1.2,-.2) {$e$};

            \draw[->] (-.8,.2) -- (-.4,.2);
            \draw[-] (.4,.2) -- (.8,.2) ;

            \draw[-] (-.8,.2) -- (-.8,.9);
            \draw[-] (.8,.2) -- (.8,.9);

            \draw[->] (.8,.9) -- (.3,.9);
            \draw[-] (-.8,.9) -- (-.3,.9);

            \node at (-1,.5) {$w$};
            \node at (1,.5) {$z$};
        \end{tikzpicture}
        \caption{LFR of $S$.}
        \label{fig:lfr_S}
    \end{subfigure}
    \hfill
    \begin{subfigure}[t]{0.32\linewidth}
     \centering
        \begin{tikzpicture}
            \draw[black] (-.4,-.4) rectangle (.4,.4);
            \node at (0,0) {$P^L$};

            \draw[black] (-.3,-.3+.9) rectangle (.3,.3+.9);
            \node at (0,.9) {$\phi$};

            \draw[->] (-1,-.2) -- (-.4,-.2);
            \node at (-1.2,-.2) {$e$};

            \draw[->] (.4,-.2) -- (1,-.2);
            \node at (1.2,-.2) {$y$};

            \draw[->] (-.8,.2) -- (-.4,.2);
            \draw[-] (.4,.2) -- (.8,.2) ;

            \draw[-] (-.8,.2) -- (-.8,.9);
            \draw[-] (.8,.2) -- (.8,.9);

            \draw[->] (.8,.9) -- (.3,.9);
            \draw[-] (-.8,.9) -- (-.3,.9);

            \node at (-1,.5) {$w$};
            \node at (1,.5) {$z$};
        \end{tikzpicture}
        \caption{LFR of $L$.}
        \label{fig:lfr_L}
    \end{subfigure}
    \caption{Interconnections for the analysis of simple feedback.}
    \label{fig:interconnections}
    \vspace{-1em}
\end{figure}

Analogously to the Nyquist criterion, we are interested in the simple feedback interconnection as shown in Fig.~\ref{fig:simple_feedback}, where $L: \Lte \to \Lte$ is the loop transfer: $y = Le$. For loop shaping, we consider both the loop transfer, and the sensitivity $e=Sr$ given by $S = (1+L)^{-1}$. To simplify the analysis, we focus on SISO systems with only one nonlinearity $\phi$ in the interconnection, which can be static or dynamic, and for which an SRG bound is available. We can write both $L$ and $S$ in LFR form using $w = \phi (z)$ and
\begin{equation*}
    \begin{pmatrix}
        z \\ y
    \end{pmatrix} = \underbrace{
    \begin{pmatrix}
        P_\mathrm{zw}^L & P_\mathrm{ze}^L \\ 
        P_\mathrm{yw}^L & P_\mathrm{ye}^L
    \end{pmatrix} }_{=: P^L}
    \begin{pmatrix}
        w \\ e
    \end{pmatrix}, \quad 
    \begin{pmatrix}
        z \\ e
    \end{pmatrix} = \underbrace{
    \begin{pmatrix}
        P_\mathrm{zw}^S & P_\mathrm{zr}^S \\ 
        P_\mathrm{ew}^S & P_\mathrm{er}^S
    \end{pmatrix} }_{=: P^S}
    \begin{pmatrix}
        w \\ r
    \end{pmatrix},
\end{equation*}
as depicted in Figs.~\ref{fig:lfr_L} and \ref{fig:lfr_S}, respectively,
which results in the operators
\begin{subequations}\label{eq:lfr}
\begin{align}
    S &= P_\mathrm{ew}^S (\phi^{-1} - P_\mathrm{zw}^S)^{-1} P_\mathrm{zr}^S + P_\mathrm{er}^S, \label{eq:lfr_S} \\
    L &= P_\mathrm{yw}^L (\phi^{-1} - P_\mathrm{zw}^L)^{-1} P_\mathrm{ze}^L + P_\mathrm{ye}^L , \label{eq:lfr_L}
\end{align} 
\end{subequations}
where all operators $P$ are SISO and LTI, and $P_\mathrm{yw}^L, P_\mathrm{ze}^L, P_\mathrm{ye}^L, P_\mathrm{ew}^S, P_\mathrm{zr}^S, P_\mathrm{er}^S$ are assumed to be stable. The stability of both operators $L,S$ depends only on $(\phi^{-1} - P^\bullet_\mathrm{zw})^{-1}$, where $\bullet\in \{S, L\}$, which can be analyzed using Proposition~\ref{thm:chaffey_thm2} by picking $H_1=\phi, H_2 = -P_\mathrm{zw}$. If $P^\bullet_\mathrm{zw}$ is unstable, one must use the SRG method for unstable systems from~\cite{krebbekxScaledRelativeGraph2024}. Period preservation of $S$ and $L$ is required for the method to work.

\subsection{Nonlinear Bandwidth}

Now that we have a definition of the NL Bode diagram, in terms of~\eqref{eqs:Gammas}, we can also define a \emph{bandwidth} (BW) for the NL loop transfer $L$ and sensitivity $S$. The following two definitions are entirely analogous to the LTI case~\cite{skogestadMultivariableFeedbackControl2010}.

\begin{definition}\label{def:closed_loop_bandwidth}
    The closed-loop bandwidth is the smallest value $\omega_B$ such that $|\underline{\Gamma}_{\omega_B}(S)|$ crosses $1/\sqrt{2}$ from below.
\end{definition}

\begin{definition}\label{def:open_loop_bandwidth}
    The open-loop bandwidth is the smallest value $\omega_c$ such that $|\overline{\Gamma}_{\omega_c}(L)|$ crosses $1$ from above.
\end{definition}

When $S$ is LTI, $\omega_B$ yields the LTI definition~\cite[Def. 2.1]{skogestadMultivariableFeedbackControl2010}. If $L$ is LTI, $\omega_c$ yields an \emph{upper bound} for the LTI crossover frequency~\cite[p. 39]{skogestadMultivariableFeedbackControl2010}, and is equal to the LTI definition if $L(s)$ crosses $1$ once, or $\Gamma_\omega$ is used instead of $\overline{\Gamma}_\omega$.

We will show later through an example why these definitions make sense and what practical use they represent.

\subsection{Frequency-Domain Analysis}

SRG analysis is used for computing the frequency-wise norms in~\eqref{eqs:Gammas} for the sensitivity $S$ and loop transfer $L$. More precisely, one can use the SRG restricted to a specific input space $\SRG_\mathcal{U}(S), \SRG_\mathcal{U}(L)$ in~\eqref{eq:srg_specific_input_space}, where now the input space is $\mathcal{U} \in \{\mathcal{U}_\omega, \overline{\mathcal{U}}_\omega, \underline{\mathcal{U}}_\omega \}$, as defined in~\eqref{eqs:input_spaces}. %

\subsubsection{The SRG of an LTI operator for a specific input space}

The first step is to consider to what input spaces different operators map to. LTI operators map each of the spaces $\mathcal{U} \in \{\mathcal{U}_\omega, \overline{\mathcal{U}}_\omega, \underline{\mathcal{U}}_\omega \}$ to itself. Because of this property, one can determine the SRG of an LTI operator $G$ with its transfer function denoted as $G(s)$ for a specific input space:
\begin{subequations}\label{eqs:lti_input_space_srgs}
\begin{align}
    \SRG_{\mathcal{U}_\omega}(G) &= \{ G(\pm j \omega)\}, \label{eq:lti_input_space_srg_sinus} \\
    \SRG_{\overline{\mathcal{U}}_\omega}(G) &= \hco(\{G(j \omega n) \mid n \in \Z \setminus \{0\} \}), \label{eq:lti_input_space_srg_harmonic}\\
    \SRG_{\underline{\mathcal{U}}_\omega}(G) &= \hco(\{G(j \omega / n) \mid n \in \Z \setminus \{0\} \}), \label{eq:lti_input_space_srg_subharmonic} \\
    \SRG(G) &= \hco(\{ G(j \tilde{\omega}) \mid \tilde{\omega} \in \R \}). \label{eq:lti_input_space_srg_full}
\end{align}
\end{subequations}
Our contribution is the SRG in~\eqref{eq:lti_input_space_srg_sinus}, \eqref{eq:lti_input_space_srg_harmonic} and \eqref{eq:lti_input_space_srg_subharmonic}, whereas the SRG in~\eqref{eq:lti_input_space_srg_full} was already derived in~\cite{chaffeyGraphicalNonlinearSystem2023}.

\subsubsection{Computing nonlinear Bode plots using SRGs}

A period-preserving nonlinear operator $R$ can generate higher harmonics, therefore for a fixed $\omega \in \R_{>0}$
\begin{equation}\label{eq:incremental_input_space_mappings}
    R : \mathcal{U}_\omega \to \overline{\mathcal{U}}_\omega, \quad R : \overline{\mathcal{U}}_\omega \to \overline{\mathcal{U}}_\omega, \quad R : \underline{\mathcal{U}}_\omega \to \Lte.
\end{equation}
Eqs.~\eqref{eqs:lti_input_space_srgs} and \eqref{eq:incremental_input_space_mappings} can be used to determine which input specific SRGs to use for the LTI parts of the LFR in~\eqref{eq:lfr_S} in the following way. First, define an SRG bound for an LFR of the form in~\eqref{eq:lfr_S} as
\begin{equation}\label{eq:lfr_srg_bound}
\begin{split}
    \SRG_{\mathcal{U}_1}(P_\mathrm{ew}^S) (\SRG(\phi)^{-1} - \SRG_{\mathcal{U}_1}(P_\mathrm{zw}^S))^{-1} \\ \times \SRG_{\mathcal{U}_2}(P_\mathrm{zr}^S) + \SRG_{\mathcal{U}_2}(P_\mathrm{er}^S) =: \mathcal{G}^\mathrm{LFR}_{\mathcal{U}_1,\mathcal{U}_2}(S),
\end{split}
\end{equation}
where $\mathcal{U}_1,\mathcal{U}_2$ are arbitrary sets. Note that~\eqref{eq:lfr_srg_bound} is evaluated using Proposition~\ref{prop:srg_calculus}, where the chord/arc properties are satisfied for sums/products of operators. Then, we can bound the SRG of $S$ for specific input spaces as
\begin{subequations}\label{eqs:srg_input_space_bound_with_lfr}
\begin{align}
    \SRG_{\mathcal{U}_\omega}(S) &\subseteq \mathcal{G}^\mathrm{LFR}_{\overline{\mathcal{U}}_\omega, \mathcal{U}_\omega}(S), \\
    \SRG_{\overline{\mathcal{U}}_\omega}(S) &\subseteq \mathcal{G}^\mathrm{LFR}_{\overline{\mathcal{U}}_\omega, \overline{\mathcal{U}}_\omega}(S),\\
    \SRG_{\underline{\mathcal{U}}_\omega}(S) &\subseteq \mathcal{G}^\mathrm{LFR}_{\Lte, \underline{\mathcal{U}}_\omega}(S), \\
    \SRG(S) &\subseteq \mathcal{G}^\mathrm{LFR}_{\Lte, \Lte}(S).
\end{align}
\end{subequations}
The sets in~\eqref{eqs:srg_input_space_bound_with_lfr} relate to the frequency-dependent gains in~\eqref{eqs:Gammas} as 
\begin{subequations}\label{eqs:freq_dep_gain_bounds}
\begin{align}
    \Gamma_\omega(S) \leq \rmin(\mathcal{G}^\mathrm{LFR}_{\overline{\mathcal{U}}_\omega, \mathcal{U}_\omega}(S)) &=: \hat{\Gamma}_\omega(S), \\ \overline{\Gamma}_\omega(S) \leq \rmin(\mathcal{G}^\mathrm{LFR}_{\overline{\mathcal{U}}_\omega, \overline{\mathcal{U}}_\omega}(S)) &= : \hat{\overline{\Gamma}}_\omega(S), \\ \underline{\Gamma}_\omega(S) \leq \rmin(\mathcal{G}^\mathrm{LFR}_{\Lte, \underline{\mathcal{U}}_\omega}(S)) &=: \hat{\underline{\Gamma}}_\omega(S), \\
    \Gamma(S) \leq \rmin(\mathcal{G}^\mathrm{LFR}_{\Lte, \Lte}(S)) &=: \hat{\Gamma}(S),
\end{align}
\end{subequations}
where the hats are used to indicate that they are \emph{upper bounds}, and not necessarily exact. By the same argument that was used for deriving~\eqref{eq:sinus_harmonic_gain_inequalities}, we can conclude
\begin{equation}\label{eq:sinus_harmonic_gain_bounds_inequalities}
    \hat{\Gamma}_\omega(S) \leq \hat{\overline{\Gamma}}_\omega(S) \leq \hat{\Gamma}(S), \quad \hat{\Gamma}_\omega(S) \leq \hat{\underline{\Gamma}}_\omega(S) \leq \hat{\Gamma}(S).
\end{equation}
Note that since $\lim_{\omega \to \infty} \underline{\mathcal{U}}_\omega = \lim_{\omega \downarrow 0} \overline{\mathcal{U}}_\omega = \Lte$, we can conclude (analogously to~\eqref{eq:harm_subharm_limit}) that 
\begin{equation}\label{eq:harm_subharm_bound_limit}
    \maketextstyle \lim_{\omega \to \infty} \hat{\underline{\Gamma}}_\omega (S) = \lim_{\omega \downarrow 0} \hat{\overline{\Gamma}}_\omega(S) = \hat{\Gamma}(S).
\end{equation}
We denote $\hat{\omega}_B$ and $\hat{\omega}_c$ as the bandwidths that are estimated using~\eqref{eqs:freq_dep_gain_bounds} in Definitions~\ref{def:closed_loop_bandwidth} and \ref{def:open_loop_bandwidth}. From~\eqref{eqs:freq_dep_gain_bounds} it is readily derived that
\begin{equation*}
    \hat{\omega}_B \leq \omega_B,\quad \hat{\omega}_c \geq \omega_c.
\end{equation*}

To summarize, one should follow the following recipe to compute the NL Bode plot.
\begin{enumerate}
    \item Write the system $S$ in LFR form to arrive at~\eqref{eq:lfr_S}.
    \item Compute the SRGs in~\eqref{eqs:srg_input_space_bound_with_lfr} using Proposition~\ref{prop:srg_calculus}. 
    \item Compute the radius to arrive at~\eqref{eqs:freq_dep_gain_bounds} and plot these values as function of the frequency.
\end{enumerate}

The analysis of the loop transfer $L$ (or any other transfer) follows exactly the same steps.

\subsubsection{Plants with integrators}\label{sec:gain_bounds_with_integrator}

As mentioned in Section~\ref{sec:lfr_simple_feedback}, we assume that $P_\mathrm{yw}^L, P_\mathrm{ze}^L, P_\mathrm{ye}^L, P_\mathrm{ew}^S, P_\mathrm{zr}^S, P_\mathrm{er}^S$ are stable and that $(\phi^{-1} - P_\mathrm{zw}^L)^{-1}$ and $(\phi^{-1} - P_\mathrm{zw}^S)^{-1}$ are stable on $\Lte$. However, when any of these LTI operators contains an integrator, which is commonplace in practice, the assumptions for our analysis would not hold. 

However, noting that $\SRG_{\overline{\mathcal{U}}_\omega}(G)$ is bounded for all $G$ with only stable poles and integrators, we can still compute $\hat{\Gamma}_\omega(S)$ and $\hat{\overline{\Gamma}}_\omega(S)$ for all $\omega$ such that these bounds return a finite value. This can be understood from the LTI case, where an integrator is not \emph{Bounded-Input-Bounded-Output} (BIBO) stable for all possible inputs, but is BIBO stable for periodic inputs with zero mean. This \enquote{extension} for plants with integrators is particularly useful for analyzing the loop transfer $L$, which commonly contains integrators. Examples of these cases are motion control setups~\cite{pogromskyFrequencyDomainPerformance2014}. Note that for this extension to work, the system must be period-preserving.

\subsection{Loop shaping}

The frequency-dependent gain bounds in~\eqref{eqs:freq_dep_gain_bounds} can be used for the design of controllers with performance guarantees. We distinguish two different approaches: loop shaping and mixed-sensitivity shaping.

\subsubsection{Interpretation of the gain bounds}\label{sec:interpretations}

In the NL case, one must use $\hat{\underline{\Gamma}}_\omega(S)$ to study the low-frequency behavior and $\hat{\overline{\Gamma}}_\omega(S)$ for the high-frequency behavior. The harmonic gain bound can also be used to provide a non-approximative upper bound for the frequency-domain analysis methods in~\cite{rijlaarsdamComparativeOverviewFrequency2017}. By studying $\hat{\overline{\Gamma}}_\omega(S)$, one addresses the question of \enquote{what is the lowest frequency in the input for which the controller has no influence (i.e. $|S| \approx 1$)?} Conversely, by studying $\hat{\underline{\Gamma}}_\omega(S)$, one addresses the question of \enquote{what is the highest frequency that can be allowed in the input to guarantee good tracking behavior?} Finally, one can use $\hat{\Gamma}_\omega(S)$ if one is interested in sinusoidal inputs specifically, for example in~\cite{pavlovFrequencyDomainPerformance2007,pogromskyFrequencyDomainPerformance2014}.

\subsubsection{Application to loop shaping}

In the loop shaping case, one uses $\hat{\overline{\Gamma}}_\omega(L)$ to tune the open-loop bandwidth to a desired level. At the same time, $\hat{\Gamma}(S)$ provides an upper bound for the \emph{modulus margin} and guarantees stability when $\hat{\Gamma}(S)$ is finite. The interpretation of performance in this NL loop shaping framework is the following: if the bandwidth is at most $\hat{\omega}_c$, then it is certain that the feedback loop is not sensitive to inputs with period $T = 2 \pi / \hat{\omega}_c$ and higher harmonics.

Perhaps the more promising approach is mixed-sensitivity shaping. In that case, one computes~\eqref{eqs:freq_dep_gain_bounds} for any desired loop transfer $T$ and tunes the controller to achieve the desired shape. Alternatively, one designs input and output LTI weighting filters $W_\mathrm{in}$ and $W_\mathrm{out}$, respectively, for the relevant loop transfer, e.g. sensitivity $S$. Then, one attaches these to the LFR in~\eqref{eq:lfr_S}, changing the LTI blocks in the LFR as $P_\mathrm{ew}^S \to  W_\mathrm{out} P_\mathrm{ew}^S , P_\mathrm{zr}^S \to  P_\mathrm{zr}^S W_\mathrm{in}$ and $P_\mathrm{er}^S \to  W_\mathrm{out} P_\mathrm{er}^S W_\mathrm{in}$. The interpretation of performance in this NL mixed-sensitivity shaping framework is the following: if $\norm{r}_\mathrm{RMS} \leq 1$, then $\norm{e}_\mathrm{RMS} \leq 1$. This result can be seen as a NL generalization of the $\Hinf$ performance concept. If $\hat{\Gamma}(W_\mathrm{out} S W_\mathrm{in}) \leq 1$, we are guaranteed that for each frequency, a sinusoidal input with $\norm{r}_\mathrm{RMS} = \sqrt{|\hat{r}_1|^2+|\hat{r}_{-1}|^2} \leq 1$ results in an output that satisfies $\norm{e}_\mathrm{RMS} \leq 1$, implying $\sqrt{|\hat{e}_1|^2+|\hat{e}_{-1}|^2} \leq 1$, which satisfies the performance specifications encoded in $W_\mathrm{in}$ and $W_\mathrm{out}$.

\subsection{Comparison with Existing Methods}\label{sec:comparison_with_existing_methods}

There are several existing methods to study the frequency-dependent gain of NL systems. The oldest and most well-known is the DF~\cite{krylovIntroductionNonlinearMechanics1947,pogromskyFrequencyDomainPerformance2014,heertjesVariableGainMotion2016}. This method is only approximate and considers only the first harmonic in the response to a sinusoidal input. The advantage is that the gain can be computed for different input amplitudes and considered phase. The DF has been extended to include all harmonics in the output~\cite{pavlovUniformOutputRegulation2006,pavlovFrequencyDomainPerformance2007} and even to consider inputs that contain harmonics~\cite{pavlovSteadystatePerformanceOptimization2013,rijlaarsdamComparativeOverviewFrequency2017}.

The main advantage of our result~\eqref{eqs:freq_dep_gain_bounds} is that it is not approximate, compared to the DF method. The sinusoidal gain $\hat{\Gamma}_\omega(R)$ can be used to reproduce the DF methods that consider sinusoidal inputs. However, $\hat{\Gamma}_\omega(R)$ considers all magnitudes of the input. To obtain input amplitude-dependent results, one should constrain $\SRG_\mathcal{U}(\phi)$ to inputs certain amplitude $\epsilon>0$, where $|x(t)| \leq \epsilon$ for all $x \in \mathcal{U}$ and $t \in \R_{\geq 0}$. This idea was heuristically explored in~\cite{chaffeyLoopShapingScaled2022}, but remains to be developed to be useful.

Whereas $\hat{\Gamma}_\omega(R)$ and $\hat{\overline{\Gamma}}_\omega(R)$ can be compared with existing methods, we must emphasize that there exists no method that considers \emph{subharmonic} inputs to the best of our knowledge. Therefore, the subharmonic gain $\hat{\underline{\Gamma}}_\omega(R)$ provides a novel way to analyze the frequency-domain behavior of the system, especially in the low-frequency regime.

For the sake of fairness, it must be mentioned that the classic DF is easier to compute than the SRG-based gain bounds in~\eqref{eqs:freq_dep_gain_bounds}. Furthermore, the proposed method in this paper considers only SISO systems with one nonlinear operator, however, using the approach in~\cite{krebbekxGraphicalAnalysisNonlinear2025}, it can be extended to the MIMO case, with multiple nonlinearities, both in the incremental and non-incremental setting.

\section{Example}\label{sec:example}

\subsection{The Controlled Nonlinear DC Motor System}

We will consider the position control of a DC motor. The equations of motion follow from Newton's second law and Kirchhoff's law with counter-electromotive force
\begin{subequations}\label{eq:dc_motor}
\begin{align}
    J \Ddot{\theta} + b\dot{\theta} &= K_\mathrm{m} i, \label{eq:dc_motor_newton}\\
    L \frac{di}{dt}+Ri + \delta \sin(\theta) &= u-K_\mathrm{m} \dot{\theta}, \label{eq:dc_motor_em}
\end{align}
\end{subequations}
where $\theta$ (\si{\radian}) is the angular position of the rotor, $u$ (\si{\volt}) is the input voltage and $i$ (\si{\ampere}) is the resulting current in the armature. The parameters in~\eqref{eq:dc_motor} are $J = 0.1 \si{\kg.\m \tothe{2}}, R = 0.96 \si{\ohm}, L = 1\si{\henry}, K_\mathrm{m}=0.2, b=1.0044 \si{\N \m \s}$, while $\delta = 0.1$ is the magnitude of a parasitic NL effect in the motor. By substituting~\eqref{eq:dc_motor_newton} into~\eqref{eq:dc_motor_em} and neglecting the NL term $\delta \sin (\theta)$, one obtains the transfer function $\theta = Gu$ given by
\begin{equation}
    G(s) = \frac{1}{s}\frac{Js + b}{(Ls+R)(Js+b)+K_\mathrm{m}^2},
\end{equation}
where $s\in \C$ is the Laplace variable. For the control configuration, we consider a standard setpoint control problem with reference $r$ (\si{\radian}) and tracking error $e = r-\theta$. Let $u = K e$ be the controller composed of a gain and lead filter
\begin{equation}
    K(s) = 5 \frac{s+1}{s/10+1},
\end{equation}
for which the LTI sensitivity $S_\mathrm{LTI} = 1/(1+GK)$ is stable.

\begin{figure}[tb]
    \centering

    \tikzstyle{block} = [draw, rectangle, 
    minimum height=2em, minimum width=2em]
    \tikzstyle{sum} = [draw, circle, scale=0.5, node distance={0.5cm and 0.5cm}]
    \tikzstyle{input} = [coordinate]
    \tikzstyle{output} = [coordinate]
    \tikzstyle{pinstyle} = [pin edge={to-,thin,black}]
    
    \begin{tikzpicture}[auto, node distance = {0.3cm and 0.5cm}]
        \node [input, name=input] {};
        \node [sum, right = of input] (sum) {};
        \node [coordinate, right = of sum] (e) {};
        \node [coordinate, above = of e] (e_out) {};
        \node [coordinate, above = of e_out] (e_out2) {};
        \node [coordinate, right = of e_out2] (e_out3) {};
        \node [block, right = of e] (controller) {$K(s)$};
        \node [sum, right = of controller] (sigma) {};
        \node [block, right = of sigma] (lti) {$G(s)$};
        \node [coordinate, right = of lti] (z_intersection) {};
        \node [output, right = of z_intersection] (output) {}; %
        \node [block, below = of lti] (static_nl) {$\phi$};
        \node [coordinate, right = of static_nl] (phi_intersection) {};
    
        \draw [->] (input) -- node {$r$} (sum);
        \draw [->] (sum) -- (controller);
        
        \draw [-] (e) --  (e_out);
        \draw [-] (e_out) --  (e_out2);
        \draw [->] (e_out2) -- node {$e$} (e_out3);
        \draw [->] (controller) -- node {$u$} (sigma);
        \draw [->] (sigma) -- node {$u'$} (lti);
        \draw [-] (lti) -- node [name=z] {$\theta$} (z_intersection);
        \draw [->] (z_intersection) |- (static_nl);
        \draw [->] (static_nl) -| node[pos=0.99] {$-$} (sigma);
        \node [coordinate, below = of static_nl] (tmp1) {$H(s)$};
        \draw [->] (z_intersection) |- (tmp1)-| node[pos=0.99] {$-$} (sum);
    
    \end{tikzpicture}
    
    \caption{Block diagram of the controlled NL DC motor in the example.}
    \label{fig:controlled_DC}
    \vspace{-1em}
\end{figure}
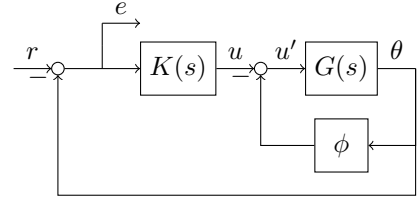

To take the nonlinearity in~\eqref{eq:dc_motor_em} into account, we define $u'=u-\phi(\theta)=u-\delta \sin(\theta)$. The resulting feedback interconnection for the controlled NL DC motor is depicted in Fig.~\ref{fig:controlled_DC}.

\subsection{Frequency-Domain Analysis of the Nonlinear Model}

To study the NL model in the frequency domain, we evaluate the right-hand side of~\eqref{eqs:freq_dep_gain_bounds} for both the sensitivity $S: r \mapsto e$ and the loop transfer $L: e \mapsto \theta$.

\subsubsection{SRG computation}

The first step is to write these operators in the LFR form of~\eqref{eq:lfr}. By setting $w = \phi(z)$, we obtain after some simple calculations that 
\begin{align*}
    P_\mathrm{zw}^L&= -G ,P_\mathrm{\theta w}^L= -G, P_\mathrm{ze}^L=L_\mathrm{LTI}, P_\mathrm{\theta e}^L= L_\mathrm{LTI}, \\ 
    P_\mathrm{zw}^S&=-S_\mathrm{LTI} G, P_\mathrm{ew}^S= S_\mathrm{LTI} G, P_\mathrm{zr}^S=S_\mathrm{LTI} G K, P_\mathrm{er}^S= S_\mathrm{LTI}.
\end{align*}
The LFR form for $S$ is also a nonlinear Lur'e system as used in~\cite{pavlovFrequencyDomainPerformance2007} with incrementally sector bounded nonlinearity. We can therefore conclude that it is a convergent system, which guarantees the period preserving property. Note that for $L$, since it is not stable by virtue of the integrator, the period preserving property has to be assumed.

Second, we need to compute the SRG of the nonlinearity. From~\cite{chaffeyGraphicalNonlinearSystem2023} we know that $\SRG(\phi) \subseteq D_{[-\delta, \delta]}$. 

To illustrate how~\eqref{eqs:freq_dep_gain_bounds} is evaluated, we show the SRG computations explicitly for some frequency values. In Fig.~\ref{fig:srg_calc_S_incr}, the necessary SRGs are plotted to evaluate~\eqref{eq:lfr_srg_bound} and compute $\hat{\Gamma}(S)$. Because all sets that are multiplied in~\eqref{eq:lfr_srg_bound} using Proposition~\ref{prop:srg_calculus}.\ref{eq:srg_calculus_parallel} have finite radius, the resulting bound $\hat{\Gamma}(S)$ must be finite. Similarly, in Figs.~\ref{fig:srg_calc_S_omega} and \ref{fig:srg_calc_L_omega} we show the SRGs that are required to compute $\hat{\overline{\Gamma}}_3(S)$ and $\hat{\overline{\Gamma}}_1(L)$, respectively, and these yield bounded results. To illustrate the problem that might occur with integrators, as discussed in Section~\ref{sec:gain_bounds_with_integrator}, we compute the necessary SRGs for $\hat{\overline{\Gamma}}_{0.05}(L)$ in Fig.~\ref{fig:srg_calc_L_problem}. Because $\SRG(\phi)^{-1}$ and $\SRG_{\overline{\mathcal{U}}_{0.05}}(P^L_\mathrm{zw})$ overlap, it holds that $0 \in \SRG(\phi)^{-1} - \SRG_{\overline{\mathcal{U}}_{0.05}}(P^L_\mathrm{zw})$, hence the SRG bound in~\eqref{eq:lfr_srg_bound} becomes unbounded.

In all cases, the SRG sums, products and inverses are computed using Proposition~\ref{prop:srg_calculus}. Additionally, chord or arc segments are added to the SRGs if required for a sum or product operation, respectively.

\begin{figure}[tb]
    \centering
    \begin{subfigure}[t]{0.24\linewidth}
     \centering
     \includegraphics[width=\linewidth]{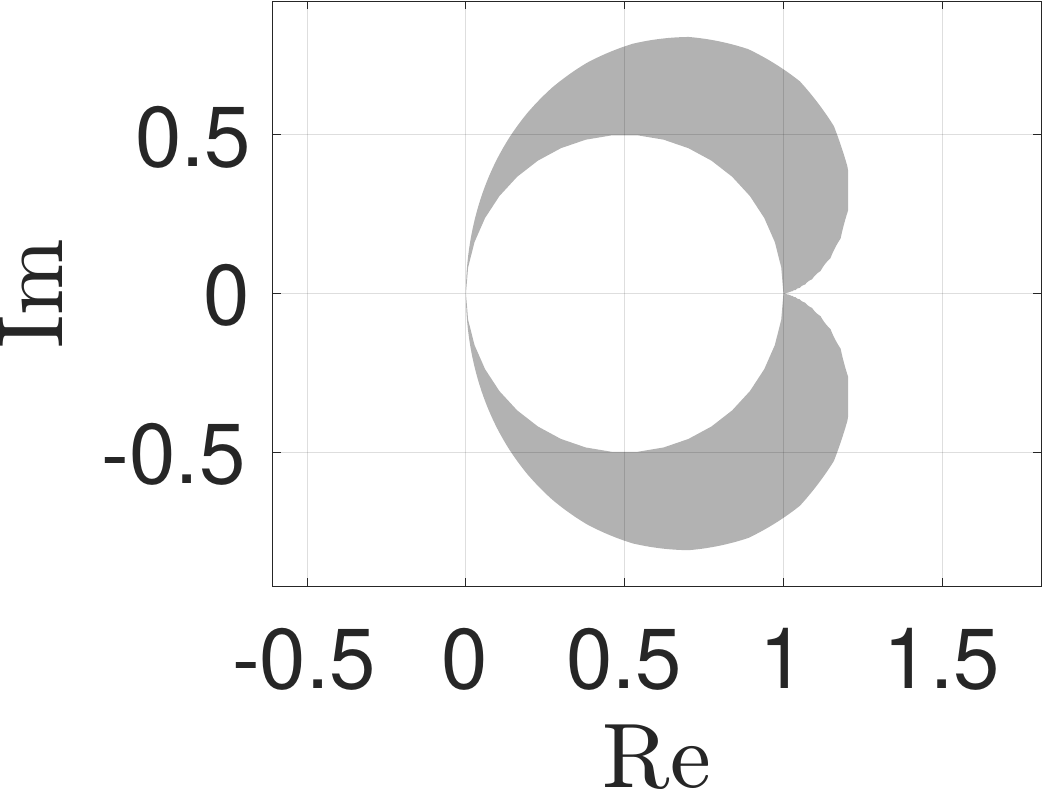}
     \caption{}
    \end{subfigure}
    \hfill
    \begin{subfigure}[t]{0.24\linewidth}
     \centering
     \includegraphics[width=\linewidth]{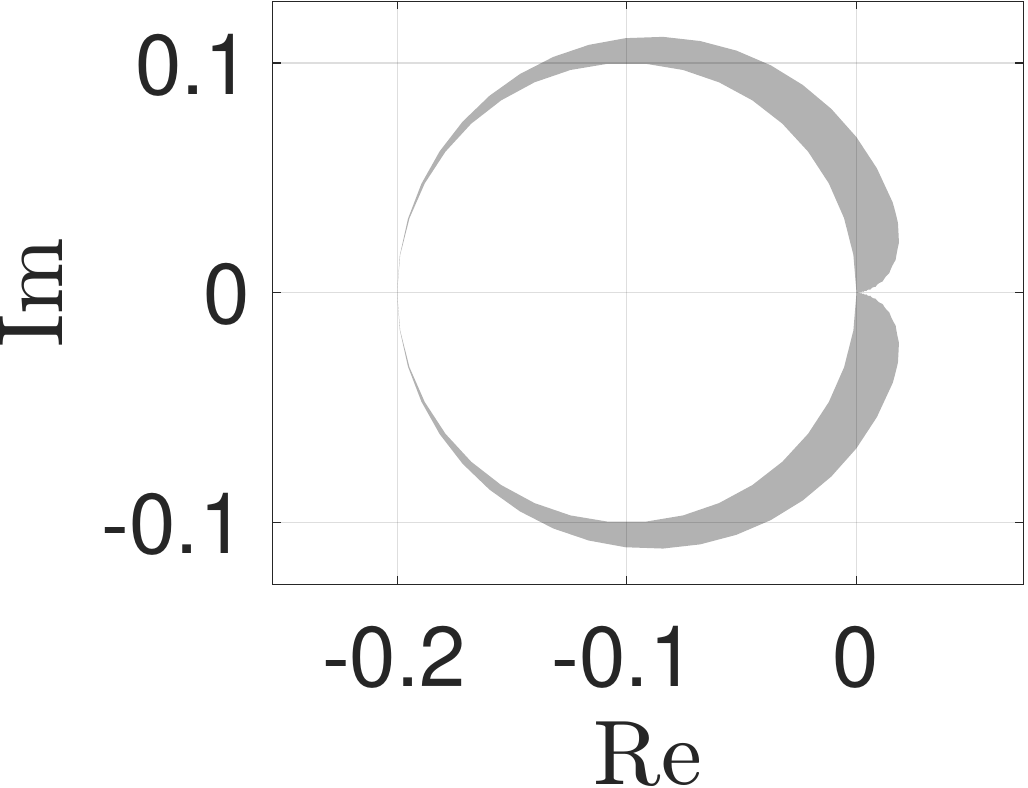}
     \caption{}
    \end{subfigure}
    \hfill
    \begin{subfigure}[t]{0.24\linewidth}
     \centering
     \includegraphics[width=\linewidth]{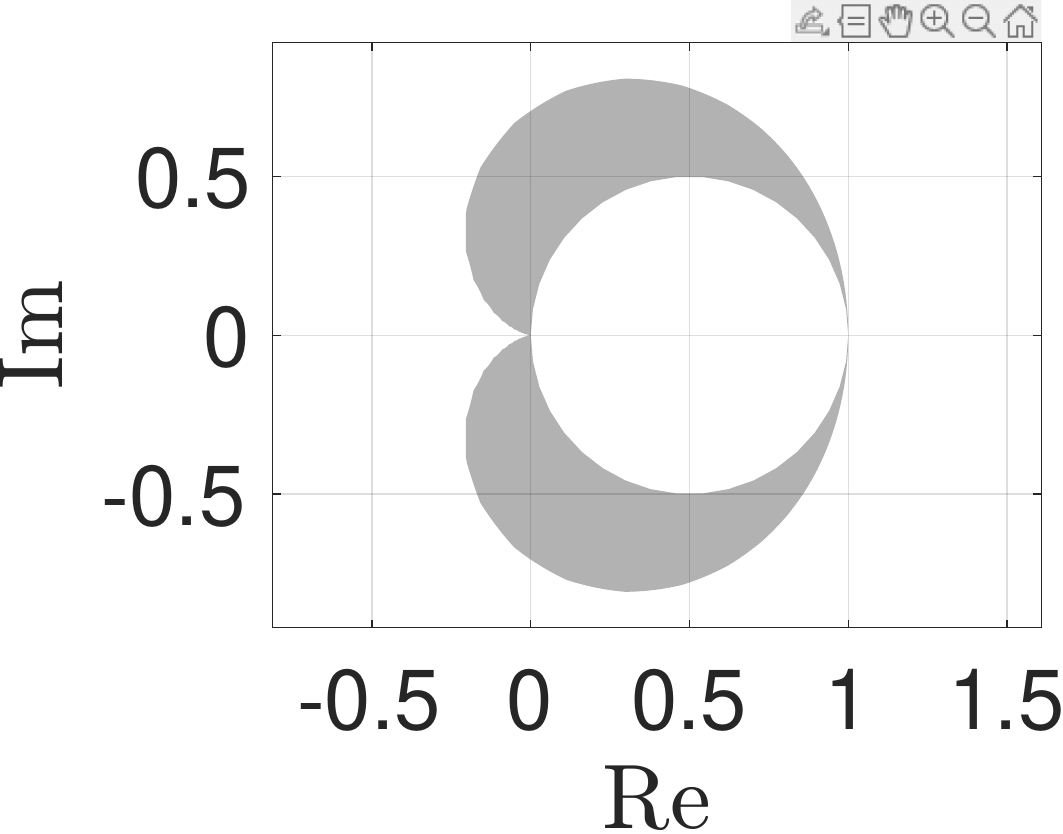}
     \caption{}
    \end{subfigure}
    \begin{subfigure}[t]{0.24\linewidth}
     \centering
     \includegraphics[width=\linewidth]{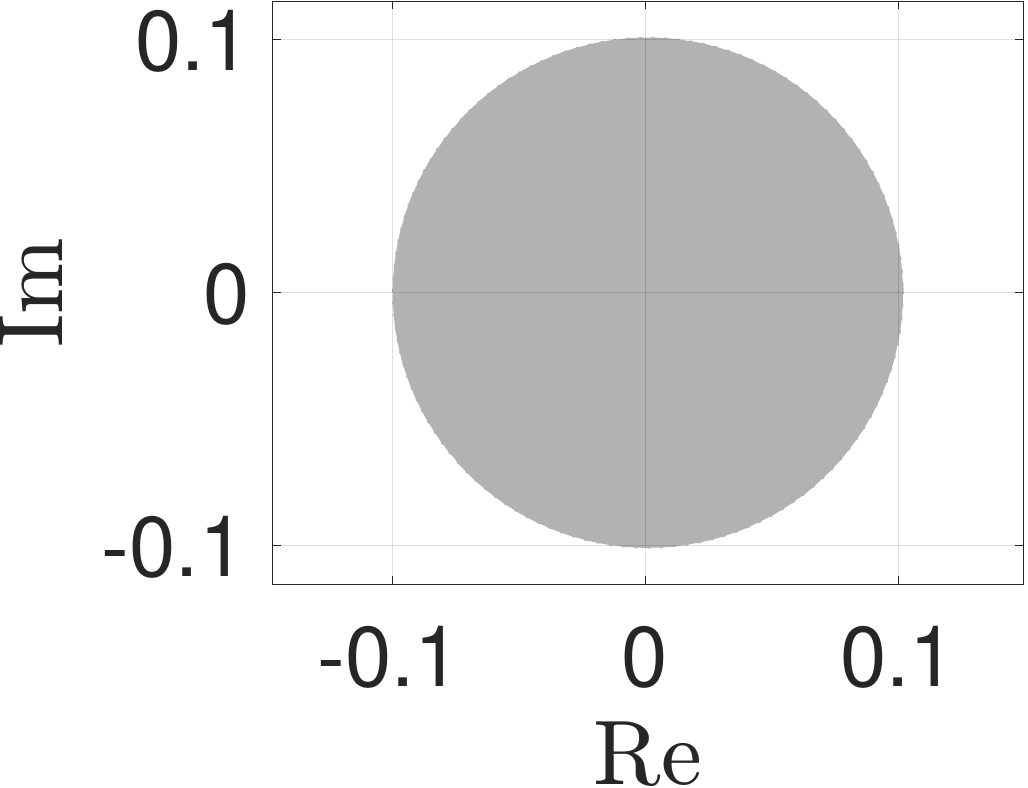}
     \caption{}
    \end{subfigure}
    
    \caption{SRGs for computing $\hat{\Gamma}(S)$. (a) $\SRG(P_\mathrm{er}^S)$, (b) $\SRG(P_\mathrm{zw}^S)$, (c) $\SRG(P_\mathrm{zr}^S)$, (d) $(\SRG(\phi)^{-1} - \SRG(P_\mathrm{zw}^S))^{-1}$.}
    \label{fig:srg_calc_S_incr}
    \vspace{-1em}
\end{figure}

\begin{figure}[t]
    \centering
    \begin{subfigure}[t]{0.24\linewidth}
     \centering
     \includegraphics[width=\linewidth]{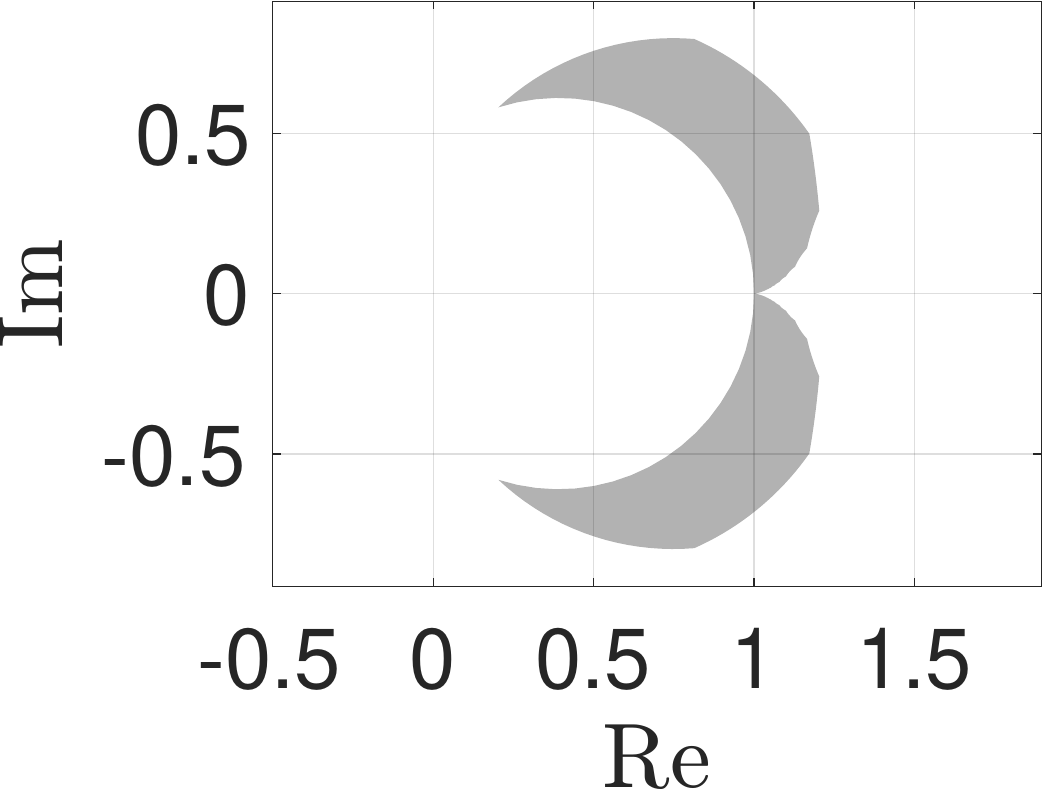}
     \caption{}
    \end{subfigure}
    \hfill
    \begin{subfigure}[t]{0.24\linewidth}
     \centering
     \includegraphics[width=\linewidth]{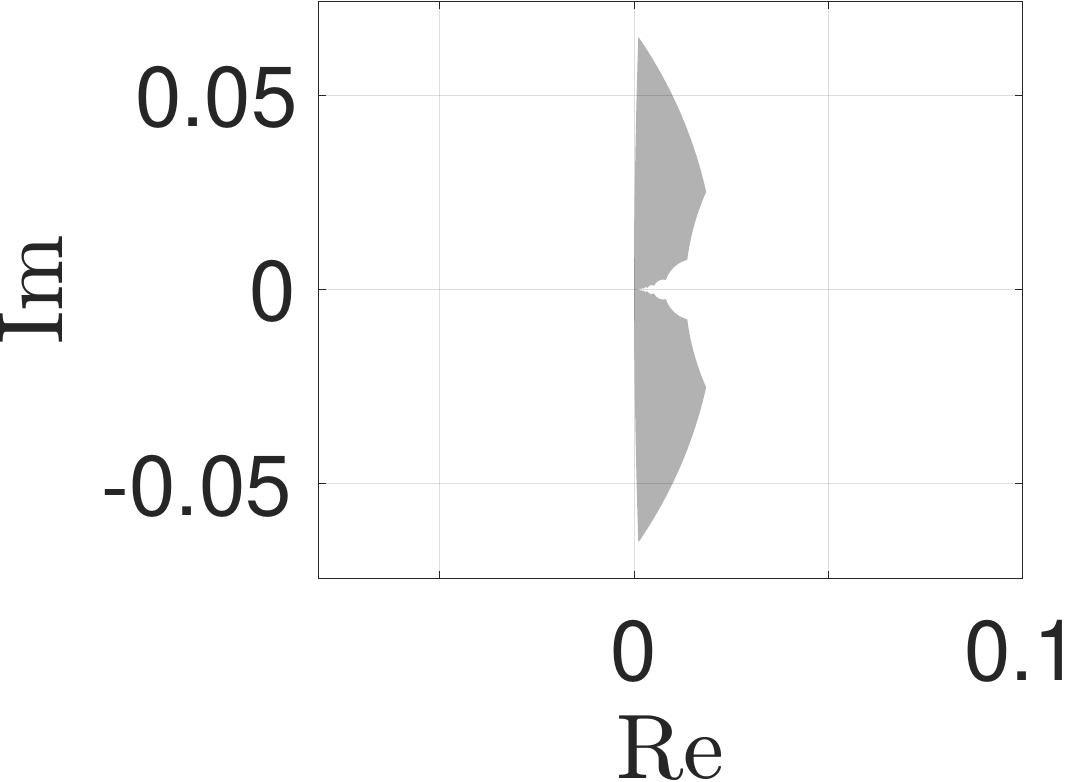}
     \caption{}
    \end{subfigure}
    \hfill
    \begin{subfigure}[t]{0.24\linewidth}
     \centering
     \includegraphics[width=\linewidth]{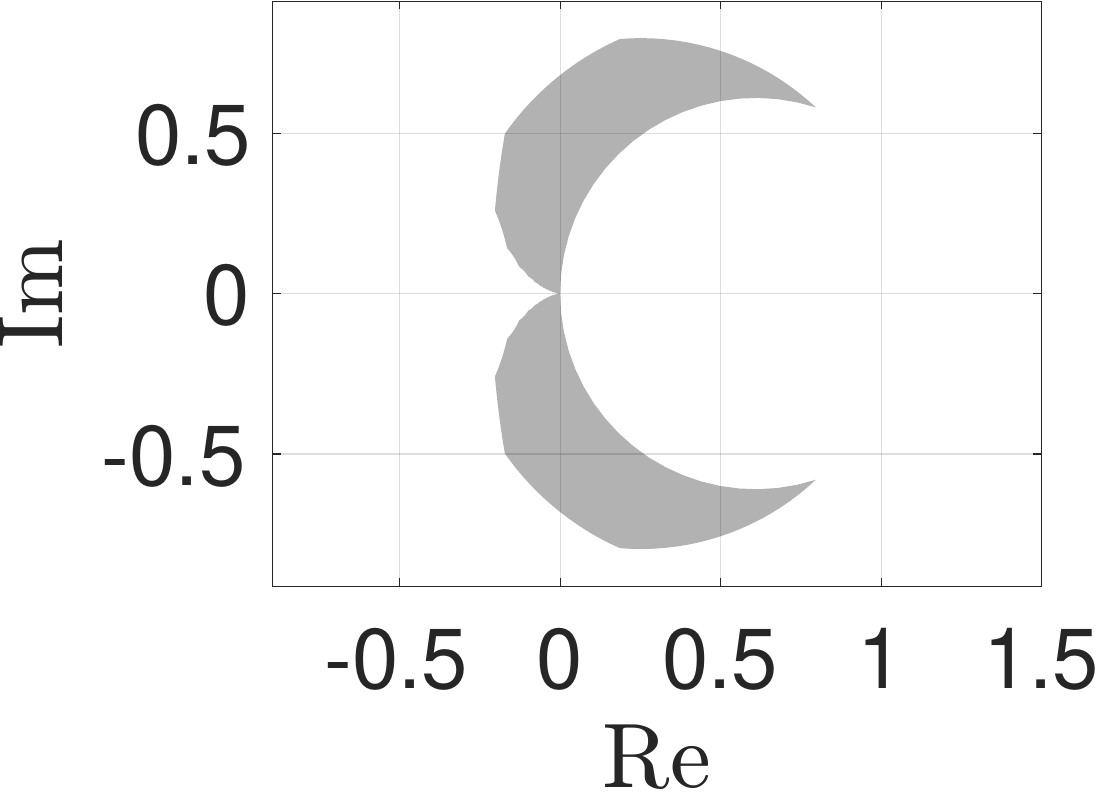}
     \caption{}
    \end{subfigure}
    \hfill
    \begin{subfigure}[t]{0.24\linewidth}
     \centering
     \includegraphics[width=\linewidth]{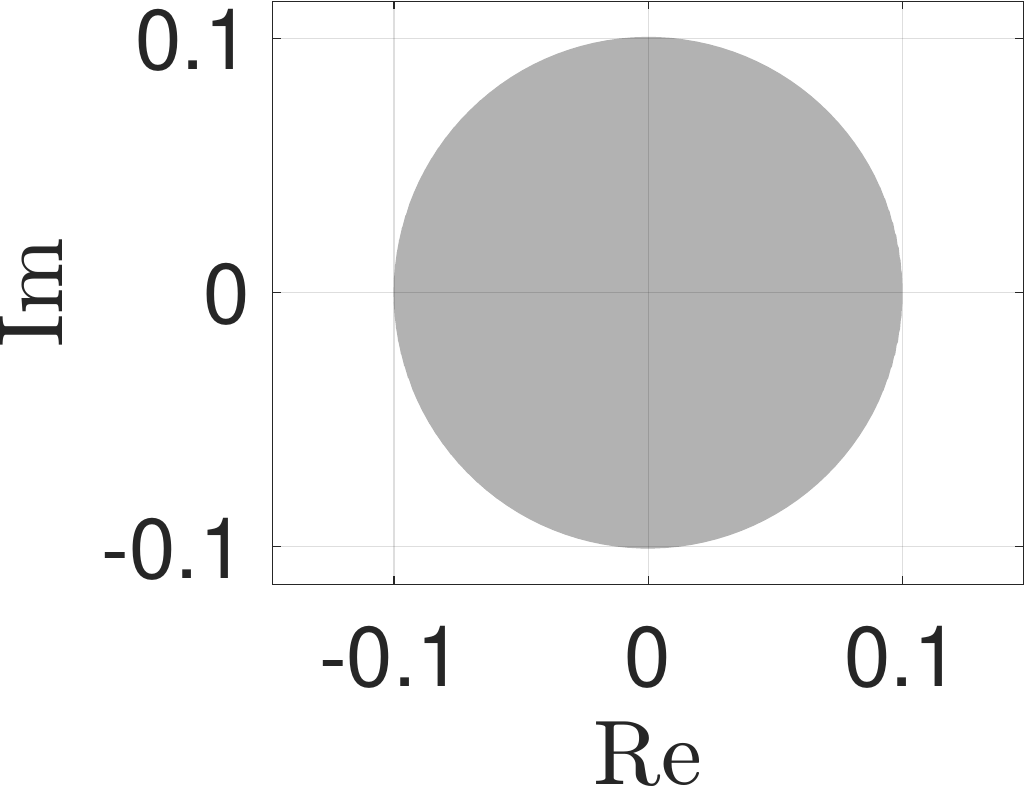}
     \caption{}
    \end{subfigure}
    
    \caption{SRGs for computing $\hat{\overline{\Gamma}}_3(S)$. (a) $\SRG_{\overline{\mathcal{U}}_3}(P_\mathrm{er}^S)$, (b) $\SRG_{\overline{\mathcal{U}}_3}(P_\mathrm{zw}^S)$, (c) $\SRG_{\overline{\mathcal{U}}_3}(P_\mathrm{zr}^S)$, (d) $(\SRG(\phi)^{-1} - \SRG_{\overline{\mathcal{U}}_3}(P_\mathrm{zw}^S))^{-1}$.}
    \label{fig:srg_calc_S_omega}
    \vspace{-1em}
\end{figure}

\begin{figure}[t]
    \centering
    \begin{subfigure}[t]{0.3\linewidth}
     \centering
     \includegraphics[width=\linewidth]{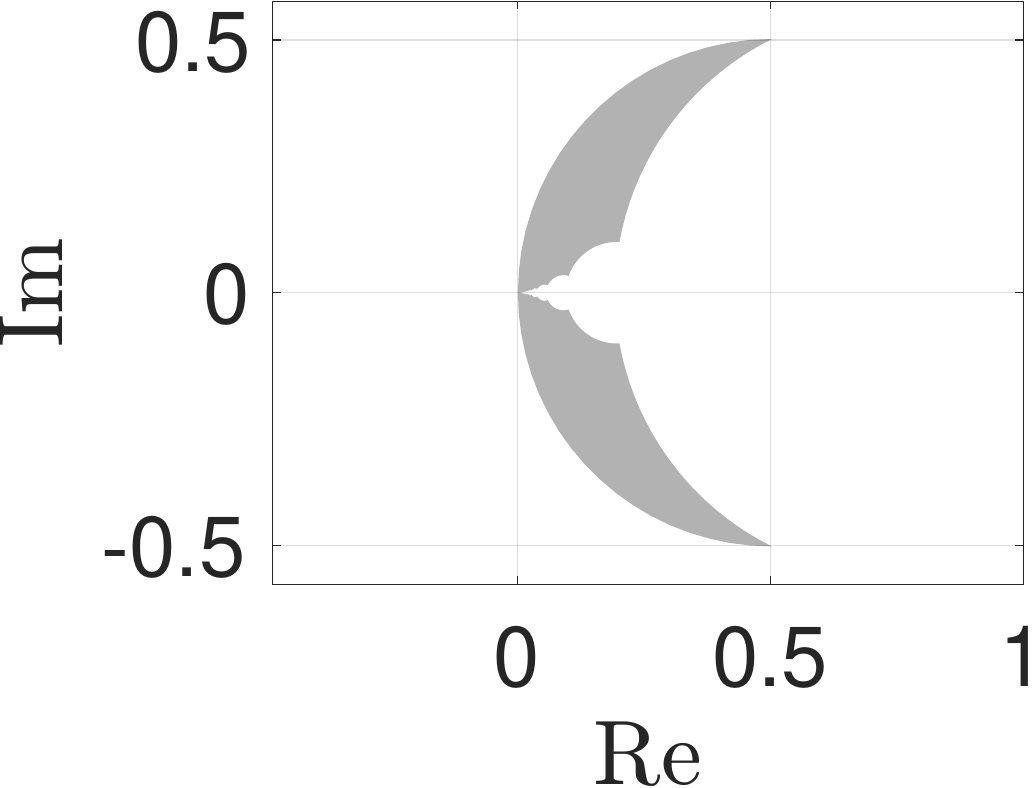}
     \caption{}
    \end{subfigure}
    \hfill
    \begin{subfigure}[t]{0.3\linewidth}
     \centering
     \includegraphics[width=\linewidth]{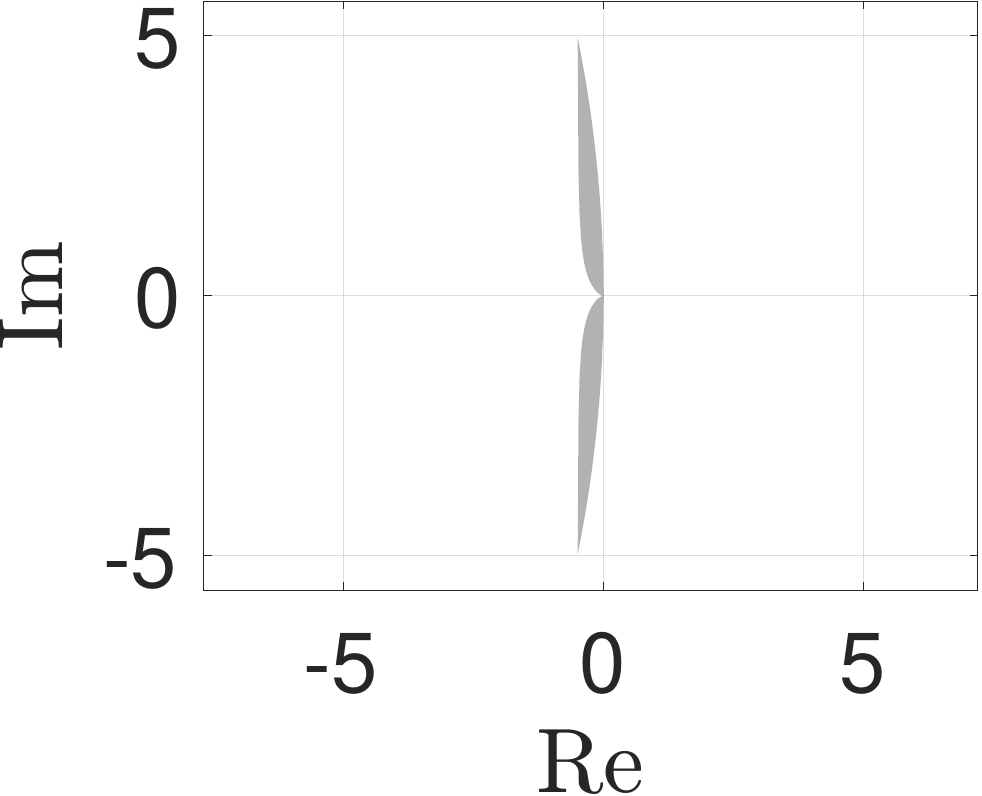}
     \caption{}
    \end{subfigure}
    \hfill
    \begin{subfigure}[t]{0.3\linewidth}
     \centering
     \includegraphics[width=\linewidth]{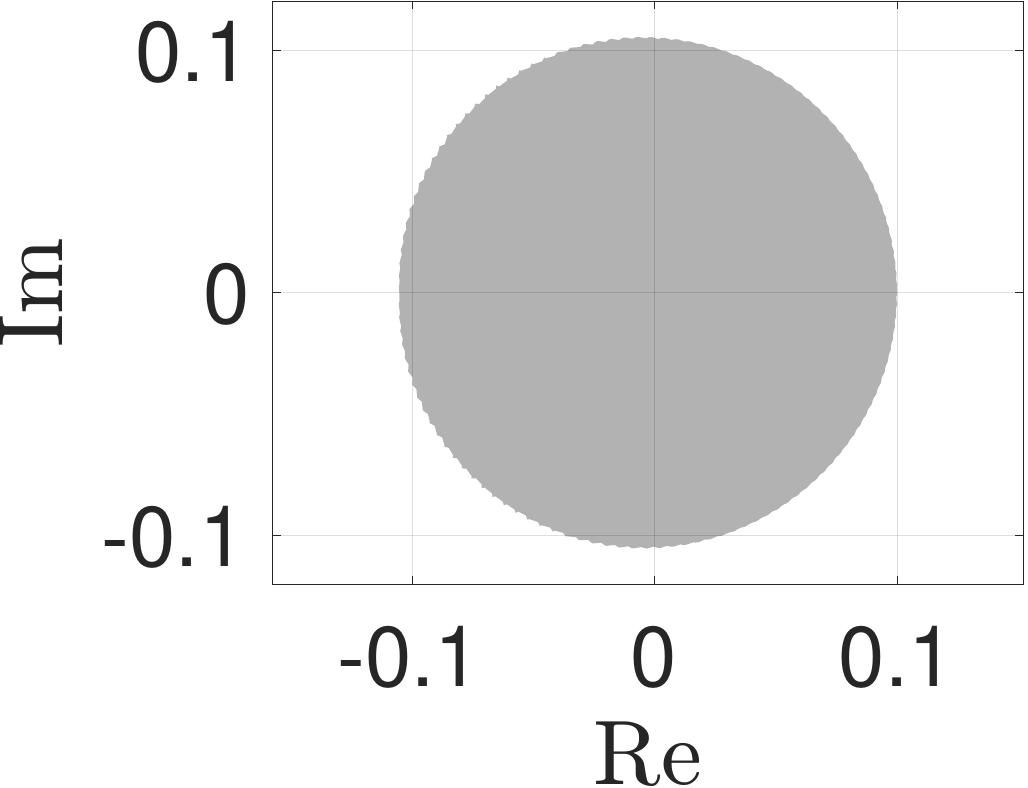}
     \caption{}
    \end{subfigure}

    \caption{SRGs for computing $\hat{\overline{\Gamma}}_1(L)$. (a) $\SRG_{\overline{\mathcal{U}}_1}(P_\mathrm{zw}^L)$, (b) $\SRG_{\overline{\mathcal{U}}_1}(P_\mathrm{ze}^L)$, (c) $(\SRG(\phi)^{-1} - \SRG_{\overline{\mathcal{U}}_1}(P_\mathrm{zw}^L))^{-1}$.}
    \label{fig:srg_calc_L_omega}
    \vspace{-1.5em}
\end{figure}

\begin{figure}[b]
    \vspace{-1em}
    \centering
    \includegraphics[width=0.4\linewidth]{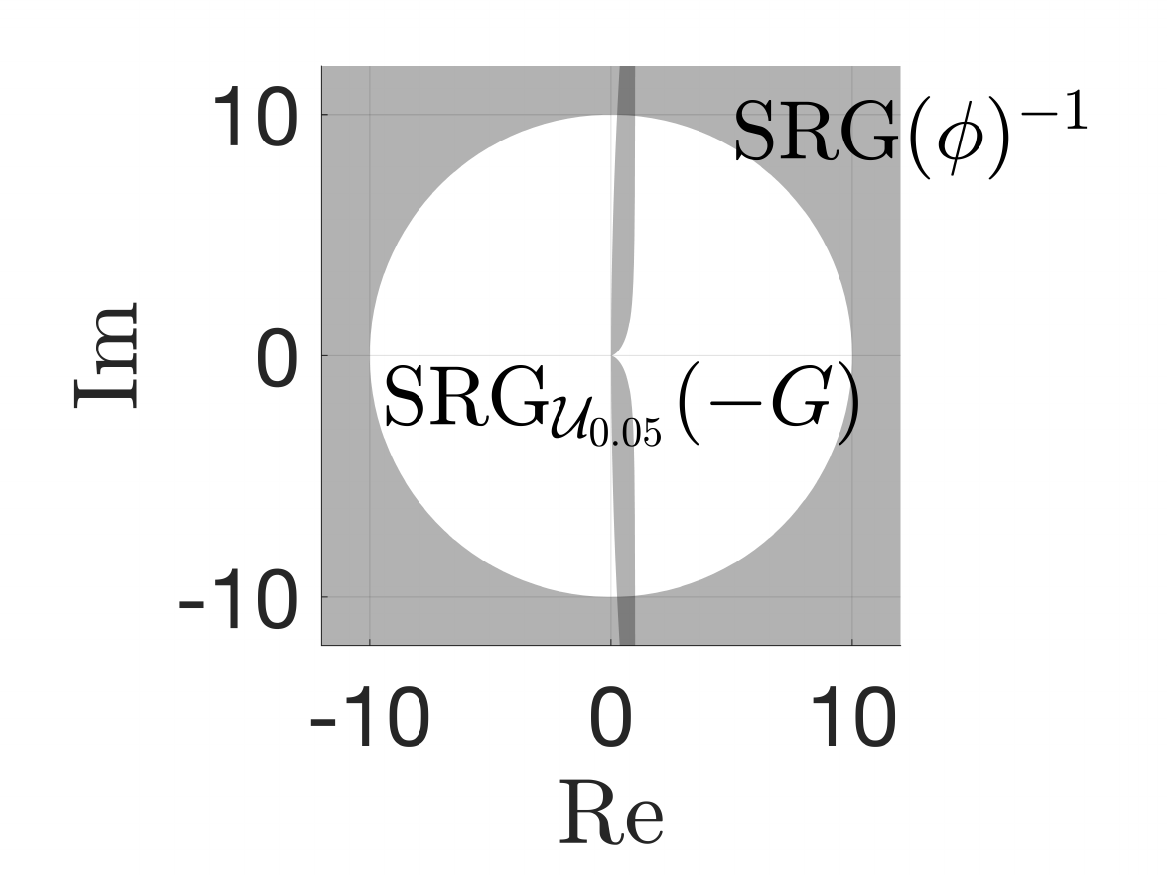}
    \caption{$\SRG(\phi)^{-1}$ and $\SRG_{\overline{\mathcal{U}}_{0.05}}(P_\mathrm{zw}^L)$}
    \label{fig:srg_calc_L_problem}
\end{figure}

\subsubsection{Sensitivity analysis}

The third and last step is to evaluate~\eqref{eqs:freq_dep_gain_bounds} on a grid of frequency points.
The result for the sensitivity $S$ is plotted in Fig.~\ref{fig:bode_S}, where we also included $S_\mathrm{LTI}$ for comparison. 

We can read off, using~\eqref{eq:harm_subharm_bound_limit}, that $\hat{\Gamma}(S) = 2.24 \si{\decibel} = 1.29$, concluding $\Gamma(S) \leq 1.29$. Now in the LTI case, one uses $S_\mathrm{LTI}$ to study both the low- and high-frequency behavior. As explained in Section~\ref{sec:interpretations}, we use $\hat{\underline{\Gamma}}_\omega(S)$ and $\hat{\overline{\Gamma}}_\omega(S)$ to study the low- and high-frequency behavior, respectively. 

From $\hat{\underline{\Gamma}}_\omega(S)$ in Fig.~\ref{fig:bode_S}, we can read off what the highest frequency is that can be allowed in the input to achieve good tracking performance. As is clear from Fig.~\ref{fig:bode_S}, the NL Bode diagram provides this information and predicts a closed-loop bandwidth estimate of $\hat{\omega}_B=3.3 \si{\radian}$. The harmonic gain bound can also be used to provide a non-approximative upper bound for the frequency-domain analysis methods in~\cite{rijlaarsdamComparativeOverviewFrequency2017}. 

Conversely, from $\hat{\overline{\Gamma}}_\omega(S)$ in Fig.~\ref{fig:bode_S}, one can read off what the response of the system is to high frequency inputs, including harmonics. The NL Bode diagram gives a frequency region in which $\hat{\overline{\Gamma}}_\omega(S)$ increases from $1$ to $\hat{\Gamma}(S)$. However, for low-frequency behavior, $\hat{\overline{\Gamma}}_\omega(S)$ is not useful.

From $\hat{\Gamma}_\omega(S)$ in Fig.~\ref{fig:bode_S} one can see what the response is to sinusoidal inputs, and that $\hat{\Gamma}_\omega(S)$ resembles a typical LTI sensitivity graph. However, the harmonic and subharmonic gain bounds offer a far more general result.

\begin{figure}[t]
    \centering
    \begin{subfigure}[t]{0.45\linewidth}
     \centering
     \includegraphics[width=\linewidth]{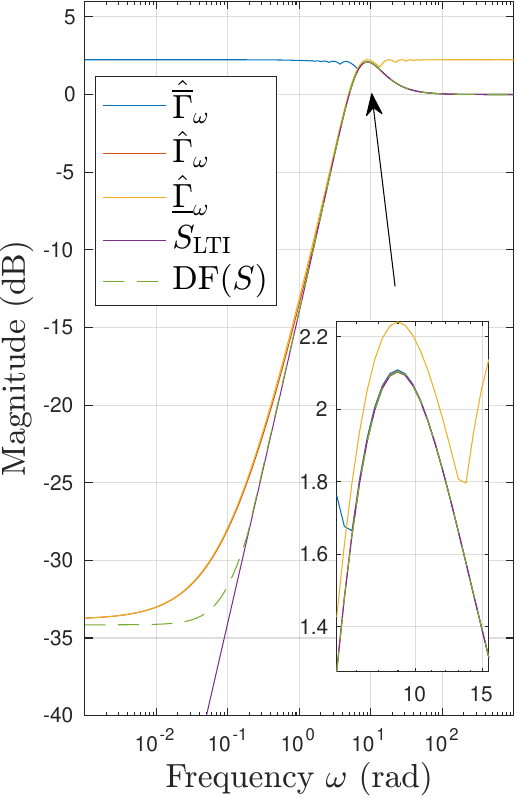}
     \caption{NL sensitivity bode diagram for the DC motor.}
     \label{fig:bode_S}
    \end{subfigure}
    \hfill
    \begin{subfigure}[t]{0.45\linewidth}
     \centering
     \includegraphics[width=\linewidth]{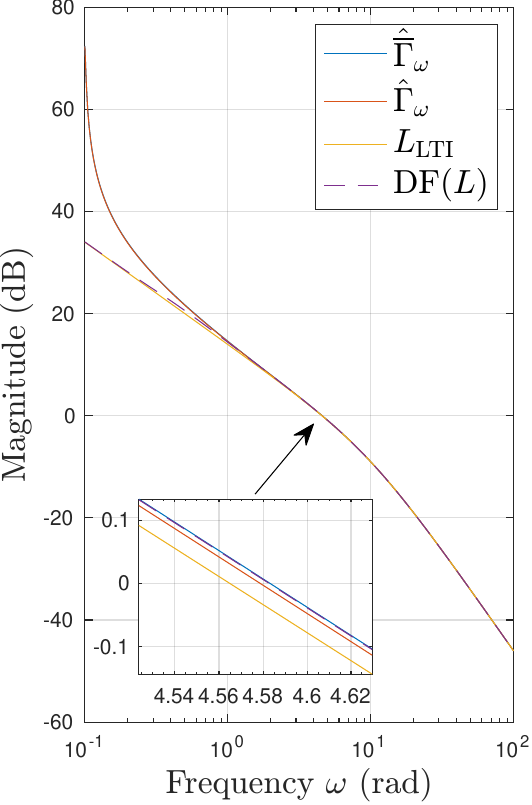}
     \caption{NL loop transfer bode diagram for the DC motor.}
     \label{fig:bode_L}
    \end{subfigure}
    \caption{NL Bode diagrams for sensitivity and loop transfer.}
    \label{fig:nl_bodeplots}
    \vspace{-1.5em}
\end{figure}

One particularly interesting feature that is present in Fig.~\ref{fig:bode_S} is the fact that the integrator behavior in $S_\mathrm{LTI}$ is no longer present in the NL case. Overall, it seems that the low-frequency behavior differs more from the LTI case than the high-frequency behavior. We will verify these observations in Section~\ref{sec:simulations}.

\subsubsection{Loop transfer analysis}

For the loop transfer $L$, we do not compute the subharmonic part from~\eqref{eqs:freq_dep_gain_bounds}. The reason for this is that the LTI parts of the LFR for $L$ contain an integrator. Therefore, as explained in Section~\ref{sec:gain_bounds_with_integrator}, the subharmonic gain will always be infinite. The relevant frequency-dependent incremental gain bounds for $L$ are plotted in Fig.~\ref{fig:bode_L}, where $L_\mathrm{LTI}$ is included for comparison.

From $\hat{\overline{\Gamma}}_\omega(L)$ in Fig.~\ref{fig:bode_L}, we can read off that $\hat{\omega}_c = 4.58 \si{\radian}$. Furthermore, we see that both the sinusoidal gain $\hat{\Gamma}_\omega(L)$ and harmonic gain $\hat{\overline{\Gamma}}_\omega(L)$ diverge at a finite nonzero value of $\omega$. The reason for this is that $\rmin((\SRG(\phi)^{-1} - \SRG_{\mathcal{U}}(P_\mathrm{zw}^L))^{-1})$ diverges due to the integrator in $G$, where $\mathcal{U} \in \{ \mathcal{U}_\omega, \overline{\mathcal{U}}_\omega \}$.

\subsubsection{Nonlinear gain upper bounds the LTI Gain}

Note that since $P^L_\mathrm{\theta e} = L_\mathrm{LTI}$ and $P^S_\mathrm{er} = S_\mathrm{LTI}$ and both $S_\mathrm{LTI}$ and $L_\mathrm{LTI}$ converge to zero, we can add $|L_\mathrm{LTI}(j \omega)| \leq \hat{\Gamma}_\omega(L)$ and $|S_\mathrm{LTI}(j \omega)| \leq \hat{\Gamma}_\omega(S)$ to the inequalities in~\eqref{eq:sinus_harmonic_gain_bounds_inequalities}. From Figs.~\ref{fig:bode_S} and \ref{fig:bode_L} one can see that these inequalities hold. 

An immediate consequence is that $\hat{\omega}_B \leq \omega_B^\mathrm{LTI}, \hat{\omega}_c \geq \omega_c^\mathrm{LTI}$, where $\omega_B^\mathrm{LTI}$ and $\omega_c^\mathrm{LTI}$ are the closed-loop and open-loop bandwidth, respectively. Note that in the LTI case, we do not need a hat to indicate the estimate since we can compute the bandwidths exactly using the LTI bode diagram.

\subsubsection{Comparison with describing function}

The DF considers only the first Fourier coefficient  of $\phi(A \sin (\omega t))$, which is denoted $B(A,\omega) \sin(\omega t + \phi(A,\omega))$, where $\phi$ is a nonlinearity~\cite{krylovIntroductionNonlinearMechanics1947}. The DF is defined as $\mathrm{N}(A,\omega) = \frac{B(A,\omega)}{A} e^{j \phi(A,\omega)}$. In the case $\phi(x) = \sin(x)$, we know from~\cite[Eq.~9.1.43]{abramowitzHandbookMathematicalFunctions1972} that $\sin(z \sin(\theta)) = 2 \sum_{k=0}^\infty J_{2k+1}(z)\sin((2k+1)\theta)$, hence $\mathrm{N}(A,\omega) = 2 J_1(A)/A$,
where $J_1$ is a Bessel function of the first kind. The DF approximation to~\eqref{eq:lfr_S} is defined as $\mathrm{DF}(S)(\omega)  = \sup_{A \in \R} | P_\mathrm{ew}^S(\omega) (\mathrm{N}(A,\omega)^{-1} - P_\mathrm{zw}^S(\omega))^{-1} P_\mathrm{zr}^S(\omega) + P_\mathrm{er}^S(\omega) |$
and $\mathrm{DF}(L)$ is defined analogously for~\eqref{eq:lfr_L}. 

They are both plotted in Fig.~\ref{fig:nl_bodeplots} and offer a gain \emph{approximation} (not a bound) between the LTI and sinusoidal ($\hat{\Gamma}_\omega$) gains, as expected. The advantage of the DF is its low computational complexity compared to the SRG method, but this comes at the cost of being only an approximation instead of a guaranteed gain bound.

\subsection{Simulation Results}\label{sec:simulations}

\begin{figure}[t]
    \centering
    \begin{subfigure}[t]{0.32\linewidth}
     \centering
     \includegraphics[width=\linewidth]{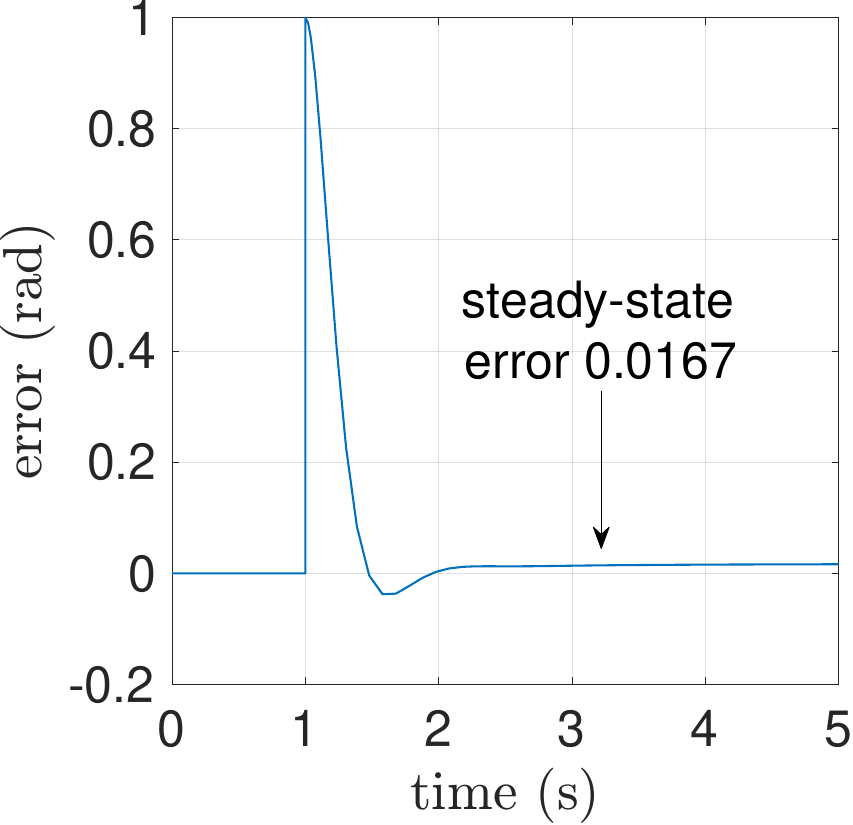}
     \caption{Step input $r_1$.}
     \label{fig:sims_dc_motor_step}
    \end{subfigure}
    \hfill
    \begin{subfigure}[t]{0.33\linewidth}
     \centering
     \includegraphics[width=\linewidth]{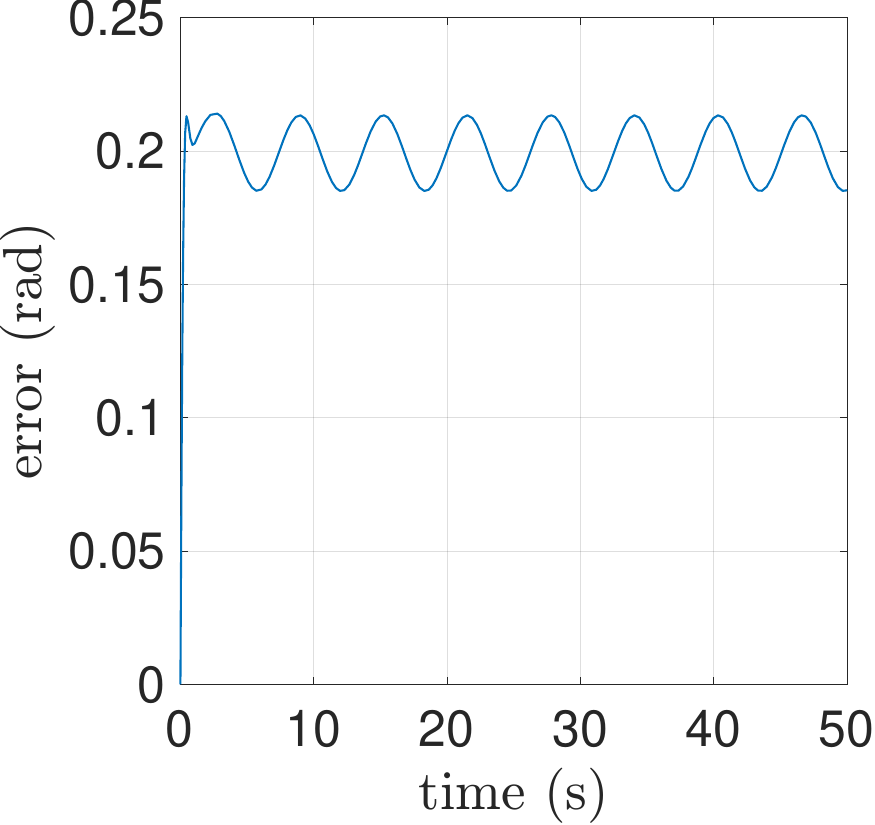}
     \caption{Ramp input $r_2$.}
     \label{fig:sims_dc_motor_ramp}
    \end{subfigure}
    \hfill
    \begin{subfigure}[t]{0.31\linewidth}
     \centering
     \includegraphics[width=\linewidth]{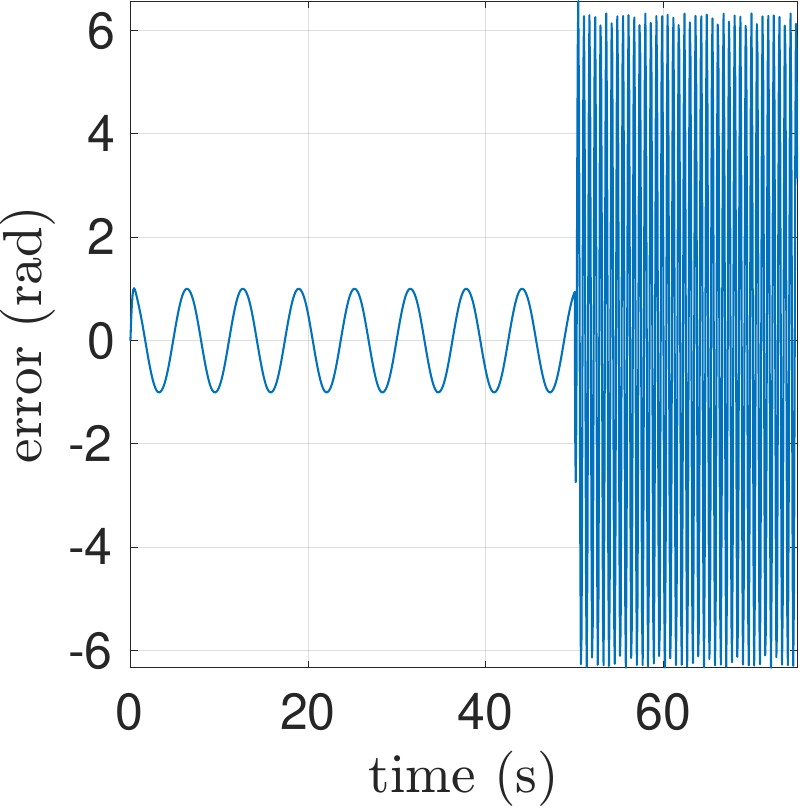}
     \caption{Periodic input $r_3$.}
     \label{fig:sims_dc_motor_periodic}
    \end{subfigure}
    \caption{Simulated responses of the controlled NL DC motor.}
    \label{fig:sims_dc_motor}
    \vspace{-1em}
\end{figure}

We simulate the system in Fig.~\ref{fig:controlled_DC} for three different references, a step $r_1$, ramp $r_2$, and periodic signal $r_3$ that switches between two different frequencies, defined by $r_2(t)=t$ and
\begin{equation*}
    r_1(t) = \begin{cases}
        0 \text{ if } t \leq 1, \\
        1 \text{ else},
    \end{cases} \,
    r_3(t) = \begin{cases}
        5 \sin(t) \text{ if } t \leq 50, \\
        5 \sin(10 t) \text{ else}.
    \end{cases}
\end{equation*}

Using the step reference $r_1$, one can see from Fig.~\ref{fig:sims_dc_motor_step} that the step response settles at a nonzero steady state error $0.0167 =-35.5 \si{\decibel}$. This corresponds to the observation that the gain $\hat{\underline{\Gamma}}_\omega(S)$ in the NL Bode diagram Fig.~\ref{fig:bode_S} has no integrator behavior. Moreover, we see that $\hat{\underline{\Gamma}}_0(\omega) \approx -34 \si{\decibel}$, which provides an upper bound for the steady state error, analogous to the LTI case.

The simulation of the ramp reference $r_2$ in Fig.~\ref{fig:sims_dc_motor_ramp} reveals a periodic response to a non-periodic input, which is a NL effect not present in the LTI model $S_\mathrm{LTI}$. The period of the response is $1 \si{\radian / \second}$, corresponding to the nonlinearity \mbox{$\propto\sin(\theta)$}, where $\theta$ tracks $r_2(t)=t$.

Finally, the simulation of the reference $r_3$, which switches between two sinusoidal signals of different frequencies, reveals three things. First, it is clear from Fig.~\ref{fig:sims_dc_motor_periodic} that the system is indeed period preserving. Second, one can read off that the amplitude gain is $1.01$ for $\omega = 1 \si{\radian / \second}$, while the amplitude gain for $10 \si{\radian / \second}$ is $6.3$. We see that $\hat{\Gamma}_1(S)\approx -13.45 \si{\decibel}=0.213$ and $\hat{\Gamma}_{10}(S)\approx 2.05 \si{\decibel}=1.27$. This corresponds to the input amplitude of $5$, since $5\cdot \hat{\Gamma}_1(S) \approx 1.06$ and $5 \cdot \hat{\Gamma}_{10}(S) \approx 6.3$, recovering the amplitude of $e$ in the steady state regimes. Note that $\hat{\Gamma}_\omega(S)$ provides \emph{upper bounds} for the amplitude, not merely approximations like the DF. Third, it is clear that there is no large transient behavior at the transition points at $t=0$ and $t=50$, which warrants the use of the NL Bode diagram to describe the performance.

\section{Conclusion}\label{sec:conclusion}

This paper develops graphical frequency-domain analysis tools for NL systems that preserve the periodicity of the input. The NL Bode diagram goes beyond existing methods that are restricted to sinusoidal inputs. In addition, we can compute the gain for subharmonic input signals, enabling a precise low-frequency sensitivity analysis. We briefly highlight how our method can be used for NL loop shaping. Our results offer a clear interpretation in the frequency domain and a definition of the NL bandwidth. Finally, the effectiveness of our method is demonstrated on the position control of a NL DC motor, and compared with the DF.

Topics of future work include the extension to MIMO systems, which includes the case of multiple nonlinearities. Additionally, it is an important objective to improve the computational efficiency of the method.

\footnotesize
\bibliographystyle{IEEEtran} 
\bibliography{bibliography} 

\end{document}